\documentclass[a4paper,12pt, notitlepage,fleqn]{article}

\usepackage{graphicx}
\usepackage[usenames, dvipsnames]{xcolor}
\usepackage{verbatim}
\usepackage[latin1]{inputenc}
\usepackage{fancyhdr}
\usepackage[T1]{fontenc}
\usepackage{epsfig}
\usepackage{subfigure}
\usepackage{amsmath}
\usepackage{amssymb}
\usepackage[amsmath,thmmarks]{ntheorem}
\usepackage{bbm}
\usepackage{mathrsfs}
\usepackage{enumitem}
\usepackage[longnamesfirst, sort]{natbib}
\usepackage{pgfplots}
\usetikzlibrary{pgfplots.groupplots}
\usetikzlibrary{external} 
\pgfplotsset{compat=newest}
\usetikzlibrary{patterns}
\usetikzlibrary{intersections}
\usepgfplotslibrary{fillbetween}
\usepackage{accents}
\usepackage{setspace}
\usepackage[hidelinks]{hyperref}

\usepackage{sgame}

\newcommand{\bmat}[1]{\left[\begin{array}{#1}}
	\newcommand{\emat}{\end{array}\right]}

\allowdisplaybreaks[1]

\newtheorem{corollary}{Corollary}

\newtheorem{proposition}{Proposition}

\newtheorem{assumption}{Assumption}
\newenvironment{proof}[1][Proof]{\noindent\textsc{#1:} }{\ \rule{0.5em}{0.5em}}

\begin{document}\onehalfspacing
	
		\title{Startup Acquisitions: Acquihires and Talent Hoarding
	}
	
	\author{Jean-Michel Benkert, Igor Letina, and Shuo Liu\thanks{Benkert: Department of Economics, University of Bern. Letina: Department of Economics, University of Bern and CEPR. Liu: Guanghua School of Management, Peking University. Email: jean-michel.benkert@unibe.ch, igor.letina@unibe.ch, shuo.liu@gsm.pku.edu.cn. We are grateful to  Heski Bar-Isaac, Florian Ederer, Johannes Johnen, Massimo Motta, Armin Schmutzler, two anonymous referees and seminar participants at the Universities of Bayreuth, Copenhagen, and Lausanne as well as at the Joint Humboldt University + University of Toronto Theory Conference in Berlin, the CCER Summer Institute in Beijing, the Swiss IO Day 2023, CEPR workshop Beyond Leveraging: Ecosystem Theories of Harm in Digital Mergers (2023), MaCCI Annual Conference 2024,  the Workshop on the Economics of Startup Acquisitions in Bern, and Swiss Society of Economics of Statistics 2024 for valuable feedback and helpful comments. Shuo Liu acknowledges financial support from the National Natural Science Foundation of China (grants 72192844 and 72322006) and Beijing Philosophy and Social Science Foundation (grant 24DTR019).}}

    \date{June 2025}
	\maketitle
	
	\begin{abstract} 

\noindent We study how competitive forces may drive firms to inefficiently acquire startup talent. In our model, two rival firms have the capacity to acquire and integrate a startup operating in an orthogonal market. We show that  firms may pursue such \emph{acquihires} primarily as a  preemptive strategy, even when they appear unprofitable in isolation. Thus, acquihires,  even absent traditional competition-reducing effects,  need not be benign, as they can lead to inefficient talent allocation. Additionally, our analysis underscores that such talent hoarding can diminish consumer surplus and exacerbate job volatility for acquihired employees.

\vspace*{0.6cm}
		

\vspace*{.1cm}

	\end{abstract}

\vspace*{.1cm}
 
	\newpage

\section{Introduction}

Historically, competition authorities were concerned with mergers and acquisitions (M\&As) only when they were likely to reduce effective competition. Since startups, almost by definition, hold small or nonexistent market shares, their acquisitions were rarely challenged \citep[e.g.][]{Bryan+Hovenkamp2020}.
But competition authorities are starting to scrutinize the effects of M\&A activity not only on current but also on \emph{potential} competition. In this context,  \cite{Cunningham2021} have shown that in the pharmaceutical industry, 5.3\%--7.4\% of all acquisitions are so-called ``killer acquisitions,'' aimed at inhibiting future competition. Against this backdrop, competition authorities believe that there may be a case for carefully scrutinizing startup acquisitions.\footnote{In 2020, the Federal Trade Commission investigated ``whether large tech companies are making potentially anticompetitive acquisitions of nascent or potential competitors'' \citep{FTCPress2020}. The European Commission has also indicated stricter enforcement \citep[see e.g.][]{EUCommissionGuidance2021}.}

This increased scrutiny is deemed unnecessary by critics arguing that startup acquisitions typically do not hamper competition -- even when they result in the killing of the startup's product or service. One common argument  \citep[e.g.][]{barnett2023killer} to support this view is that such acquisitions are so-called ``acquihires.''
As the name suggests, acquihires are essentially a hiring instrument: the acquiring firm is primarily interested in hiring the startup's employees, not removing a potential competitor. Consider the case of Drop.io, a startup offering easy file sharing. In 2009, Drop.io was a successful startup, having been named a winner of CNET's Webware 100 award and listed among the 50 best websites by \emph{Time} magazine. After acquiring Drop.io in 2010, Facebook promptly terminated it and announced that its CEO, Sam Lessin, would be assigned to a new role \citep{Webware100,Time,Mashable}. While the startup was ``killed,'' the motivation for doing so is different from the killer acquisitions of \cite{Cunningham2021}. Yet, does that necessarily mean the acquisition was benign?

The goal of our paper is to contribute to this discussion by presenting a simple yet general framework allowing for the study of acquihires. We consider a model in which two symmetric incumbents are competing in one market while a startup is operating in an orthogonal market. This rules out the elimination of potential competitors as the motivation for an acquisition. The incumbents can attempt an acquihire -- that is, acquire the startup and integrate its employees into own operations. An acquihire leads to an efficiency gain for the acquiring firm so that profits of the acquirer increase while those of the competitor decrease. We allow for different degrees of efficiency gains by modeling two levels of match quality between the incumbents and the startup employees.

We present three main results. First, we show that inefficient acquisitions may occur even if the startup is not a potential competitor to the incumbents. The inefficiency is manifested through \textit{talent hoarding}:  firms are engaging in acquihires even when it leads to lower aggregate profits for the startup and the acquirer afterward. Essentially, we find that a low-match firm can increase its expected profits by acquiring the startup before the (potentially high-match) competitor learns of its existence. Such acquihires generate an inefficiency because the startup employees would be more productive by either staying with the startup or (if the option is available) moving to the high-match competitor. Our model thus suggests that startup acquisitions need not be benign even when potential-competition motives are ruled out.

Talent hoarding in our model is driven by the preemption effect, with firms acquiring talent partly to prevent competitors from becoming stronger. Many papers have identified the preemption effect in various settings, from patent races \citep{gilbert1982preemptive}, to M\&As in the international-trade setting \citep{norback2004privatization, norback2007investment}, to technology acquisitions \citep{bryan2020antitrust}, to the case of duopolists facing capacity constrained suppliers in a decentralized market \citep{burguet+Sakovics2017}. Indeed, the point of our paper is to highlight that preemptive motives can be at work even when a startup operating in a different market is being acquired. 

However, preemptively acquiring labor is different -- and can be more pernicious -- from acquisitions of physical assets or technology, because firms do not acquire property rights over labor. In particular, when a firm preemptively obtains physical assets or technology that could be more profitably used by the competitor, both firms have an incentive  to negotiate an efficient allocation of property rights, for example by licensing the technology, as in \cite{katz1987r}. This is not the case when a firm preemptively obtains labor. In this case, the best that the firm  can do is to try to keep the workers from joining the competitor, thereby leading to an inefficient allocation of resources. A similar misallocation of labor can also occur \emph{within} a firm (and therefore without a competition motive), as \cite{haegele2022talent} has documented. Moreover, additional issues, such as the impact of the business cycle, which we will discuss later, appear.

Our second main result concerns the effect of acquihires on consumer surplus. While increasing the efficiency of the acquirer, an acquihire leads to a loss of surplus induced by the disappearance of the startup. Assuming that part of the efficiency gain is passed on to consumers, the overall effect of the acquihire on consumer surplus is ambiguous. In particular, if high-match acquisitions increase it while low-match acquisitions decrease it,  competition authorities who cannot identify match quality may face complex challenges in regulating acquihires. For instance, the welfare effects of banning acquihires might not depend straightforwardly on the ex-ante likelihood of a high-match deal. As we show, prohibiting acquihires tends to decrease consumer welfare when that probability is either very high or very low. When the probability of a high-match deal is high, low-match firms do have a strong incentive to hoard talent but most of the acquihires are with high-match firms, which on average improves welfare. When the probability of a high-match deal is low, low-match firms \emph{endogenously} choose not to engage in (expensive) talent hoarding, so all observed acquihires are with high-match firms and thus welfare enhancing. It is when the probability of a high match is intermediate -- such that low-match firms remain tempted to hoard talent and are not rare -- that banning acquihires has the largest scope for enhancing expected consumer welfare.

An important question is the extent to which the emergence of talent hoarding and its welfare implications hinge on our modeling assumptions. To this end, Section \ref{sec:robustness} considers several variations of our model including changes to the order and timing of moves, bargaining power, or the market's competitiveness. The key takeaway is that talent hoarding is not an artifact of our model but that the preemption motives and their adverse effects are a robust phenomenon. However, the section's results also highlight that there are situations in which inefficient acquihires are less of a concern and do not necessarily warrant regulatory attention. In particular, sequential bidding for the startup is important for the negative welfare effects to emerge. If the bidding is simultaneous, or if the order of bids is uncertain, the prevalence of talent hoarding is reduced. Similarly, if the startup enjoys a strong bargaining position, the adverse welfare effects are less likely to materialize. Market structure also plays a role. Talent hoarding is most likely to occur when there are few competitors and especially if there is a dominant firm facing a challenger. It is less likely to occur in a market with a monopoly firm and a potential entrant, or a market with many competitors. Overall, our results suggest that acquihires may be a concern for regulators but perhaps not warrant the same level of scrutiny as more conventional cases.  

Our final main result is that the labor-market outcomes (hiring and layoffs) for acquihired employees may become more volatile because of firms' talent hoarding. To obtain this result, we expand our baseline model by adding a second period. In between periods, the economy may fall into a recession and, consequently, firms may get hit by potentially correlated adverse shocks. We show that relative to a benchmark case in which firms have no motive to hoard talent,  talent hoarding always leads to more hiring and may also lead to more layoffs  for acquihired employees when the adverse shocks are sufficiently likely or sufficiently positively correlated. This finding lends support to the view that talent hoarding was a contributing factor to the substantial number of layoffs in the tech industry following the increased downward pressure on the US economy in 2022 \citep{NYT}.

We extend our framework in two important directions in Section \ref{sec:extensions}. First, we consider the situation where a startup holds value for firms due to both its employees and its technology. This extension highlights the insight that, unlike labor, technology can be sold or licensed following an acquisition since firms do not obtain property rights over labor. Second, we examine partial acquisitions or investments in startups. In these instances, the acquiring firm purchases only a portion of the startup, gaining a share of its profits and some decision-making authority. We demonstrate that when this decision power includes the ability to block acquisitions by competitors, such investments can be an attractive alternative to full acquihires. We also show that partial acquisitions generate lower inefficiencies than full acquihires but that they occur more often. 

We close the introduction by situating acquihires within the broader context of labor market dynamics and hiring practices. Talent hoarding, as described in our model, is most likely to arise when firms can identify high-impact teams before their competitors. This mechanism could, in principle, extend beyond acquihires to other forms of hiring, such as direct recruitment of employees with scarce and desirable skills. However, we focus on acquihires but not hiring in general because they present a particularly acute setting for talent hoarding: startups provide a clear and verifiable signal of entrepreneurial ability, while the individuals behind them may not yet be widely recognized within the industry. That said, firms may also choose to bypass the acquisition process altogether by directly poaching a startup's employees \citep[see][and the pertaining discussion in the literature review]{bar-isaac2023acquihiring}. Indeed, this practice has recently been documented in the context of AI startups and been portrayed as a way to avoid regulator scrutiny that would ensue with an acquisition of the startup \citep{reverse-aquihire, verge}. For example, in the case of Microsoft poaching employees from the AI startup Inflection, the Competition and Markets Authority in the UK investigated the transaction and treated it as a merger. Such hiring events can have similar economic and competitive implications to acquihires. Indeed, as the Microsoft/Inflection case suggests, these events may affect the acquirer, the hired employees, and the startup in ways comparable to an acquihire, warranting equivalent regulatory treatment \citep{inflection}.

\paragraph{Related literature.} Our paper is most closely related to the literature studying the economics of startup acquisitions. Much of the early literature examined, in various settings, how the prospect of an acquisition affects the incentives of startups and incumbents to invest in innovation \citep[e.g.,][]{Mason2013, norback2012, gans2000incumbency, phillips2013, Rasmusen1988}. Following \cite{Cunningham2021}, who demonstrated that incumbents may acquire startups for anticompetitive reasons, a large literature has studied the effects a more restrictive merger policy would have on innovation and overall welfare \citep{motta2021big, cabral2020merger, cabral2023big, katz2020bigtech, letina2023killer}. 
Others examine how  acquisitions can steer the direction of innovation \citep{bryan2020antitrust, callander2021novelty, dijk2021start}, and \cite{fumagalli2020shelving} consider the impact of financial constraints. 
Several papers consider dynamic incentives \citep{cabral2018standing, bryan2020antitrust, hollenbeck2020, denicolo2021acquisitions}. A key insight is that if the incumbent pulls too far ahead in the technology space, the pace of innovation will go down. There is also the possibility that the incumbent creates a kill-zone that disincentivizes entry, either by acquiring entrants, copying their products, or investing heavily in innovation \citep{kamepalli2021kill, teh2022acquisition, shelegia2021kill, bao2023killer}. On the empirical side, \cite{ederer2023great} show that startups increasingly favor acquisitions over IPOs as exit strategies. Finally, several papers empirically study acquisitions in the tech sector \citep{affeldt2021big, affeldt2021competitors, prado2022, jin2023, gugler2023start, eisfeld2022entry, gautier2020mergers}. 

Our paper differs from this literature by considering startups that are not potential competitors of the incumbents. The main channel through which acquisitions create inefficiencies is thus fundamentally different. In particular, the welfare loss does not occur due to the loss of competition, but rather due to inefficient allocation of talent. For regulators, this implies that it is not sufficient to show that there is no (potential) product market overlap between the target and the acquirer for the acquisition to be benign -- as it would be if only the concerns from the literature on potential competitors were present. Finally, while in the potential competitor model the anticompetitve concerns are typically strongest when there is a single incumbent, in our model this is not the case. At least two incumbents are needed in our model for the inefficiency to arise, and, as we show later on, the incentives to hoard talent can initially even increase as the number of incumbents grows beyond two.  

We do not examine why firms engage in acquihires instead of directly poaching valuable employees. This question is tackled by \cite{bar-isaac2023acquihiring}, who show that an acquihire can increase the monopsony power of the acquirer by removing the most relevant labor-market competitor. This in turn lowers wages, making acquihiring more profitable than direct hiring.  \cite{Coyle+Polsky2013} argue that firms engage in acquihires for reputational reasons, while \cite{selby2013startup} add that acquihiring is a method of acquiring  entire teams.

Our paper also relates to the empirical literature that directly studies acquihires. \cite{ouimet2020acquiring}, \citet{Ng+Stuart2021}, \citet{chen2021human}, and \citet{chen2022hiring} show that acquiring talent is indeed an important motivation for acquisitions. However, acquihired employees separate at a higher rate than regularly hired employees \citep{Ng+Stuart2021, verginer2022impact}, possibly due to a preference for working at startups or a misalignment with the acquirer's plans \citep{Kim2020startup, loh2019disruption}. This empirical finding is consistent with our theoretical result.

More broadly related is \citet{haegele2022talent}, who finds evidence of talent hoarding by managers \emph{within} firms. We identify strategic motives for talent hoarding across firms. The literature on endogenous technological spillovers caused by workers' changing jobs is also broadly related. The possibility that workers might move to the competitor influences whether multinational enterprises export or produce locally \citep{fosfuri2001foreign} and how much firms may invest in innovation \citep{gersbach2003endogenous_a,gersbach2003endogenous_b}. Also broadly related is the concept of \emph{labor} hoarding from macroeconomics, which refers to firms' employing more workers during economic contractions than is necessary for production. The firms do this to avoid incurring hiring and training costs once the economy recovers \citep[for an overview, see][]{Biddle2014}. Our model predicts that \emph{talent} hoarding implies more volatile hiring and firing decisions during economic expansions and contractions, which dampens the observed \emph{labor} hoarding during contractions. This is exactly what \citet[pp. 209-210]{Biddle2014} reports has been happening recently, especially during the Great Recession. If  talent hoarding has become more common, then our model provides a potential explanation for this observation.

\vspace{-1mm}
 
\section{Model} \label{sec:model}

Two symmetric firms $i\in \{1,2\}$ are competing in a market.\footnote{In many relevant applications there will be a dominant firm in the market. We discuss this extension in Section \ref{sec:conclusion}, where we also consider the effect of more than two firms.} There is a second market in which entrepreneur $E$'s startup is active. In the status quo, the firms' payoffs are given by $\Pi_F$ and the entrepreneur's payoff is $\pi_E$. 
Our model does not specify any direct linkage between the two markets (e.g. through consumer demand), as we prefer to consider them as orthogonal to each other. This allows us to rule out conventional competition motives for the firms when acquiring the startup, as will become clear later.

A firm can engage in an ``acquihire,'' whereby it acquires and integrates the startup by making a bid $p$ to the entrepreneur. If successful, the payoff consequences of the transaction depend on the match quality $\theta\in\{H, L\}$ between the acquirer and the startup. This match quality is the acquirer's private information, and it is drawn i.i.d. for each firm according to $\Pr(\theta=H)=1-\Pr(\theta=L)=\lambda\in (0, 1)$. Specifically, if firm $i$ with match $\theta_i$ successfully pursues an acquihire at bid $p$, its payoff is $\bar{\Pi}_F^{\theta_i} -p$, while the other firm's payoff is $\underline{\Pi}_F^{\theta_i}$ and the entrepreneur receives $p$.

We assume that an acquihire leads to an efficiency gain over the competitor.
\begin{assumption} \label{ass:A1} We assume that
	\begin{enumerate}[label=(\roman*)]
		\item $\bar{\Pi}^H_F>\Pi_F+ \pi_E>\bar{\Pi}_F^L$
		\item $ \Pi_F\geq \underline{\Pi}_F^L >\underline{\Pi}_F^H$
		\end{enumerate}
\end{assumption}
According to Assumption \ref{ass:A1}(i), the joint profits of the startup and the acquirer are highest when a high-match firm acquires the startup, second highest when the would-be acquirer does not acquire the startup, and lowest when a low-match firm acquires the startup. Assumption \ref{ass:A1}(ii) says that the profits of the non-acquiring firm are highest when its competitor does not engage in an acquihire, followed by when a low-match competitor engages in an acquihire and lowest when a high-match competitor does so. 

From the consumers' point of view, there are three possible outcomes. First, absent an acquihire all three firms are active in their respective markets. In this case, the consumer surplus arising from the competition between the two symmetric firms is $CS_F$ and that from the startup is $CS_E$. Second, a low-match acquihire results in competition between the two (now asymmetric) firms generating the entire consumer surplus ($CS_L$). Third, a high-match acquihire results in a similar competition that generates the entire consumer surplus ($CS_H$). Whenever an acquihire occurs, the consumer surplus generated by the startup ($CS_E$) in its original market is lost, although the market itself need not disappear entirely.\footnote{\label{fn:zenly}The reduction in consumer surplus due to the startup's disappearance can be substantial, even if the startup's profits are not. For an example of such a ``loved but unprofitable'' startup, consider the mapping app Zenly, which had more than 15 million daily users when it was controversially shut down by its parent company based on the ``low likelihood of the app ever turning a meaningful profit'' \citep{sifted}.} The next assumption captures the idea that as one of the two firms becomes more efficient, it passes on some of that efficiency gain to consumers.\footnote{While it is possible that this assumption does not hold, extending our analysis to those cases is straightforward but would come at the cost of more complex exposition.}

\begin{assumption} \label{ass:A2} Let $CS_H \geq CS_L \geq CS_F$. 	
\end{assumption}

Finally, the timing is as follows. At the outset, nature draws the private match types of the firms. In the first stage, firm 1 has the opportunity to attempt an acquihire. The entrepreneur can accept or reject the bid. If the entrepreneur accepts the bid, the game ends. Otherwise, we move to the second stage. In this stage, firm 2 has the opportunity to attempt an acquihire. The entrepreneur can accept or reject the bid, after which the game ends.\footnote{We discuss in Section \ref{sec:robustness} that the emergence of incentives to hoard talent does not hinge on firms' knowledge of the order of moves nor on the fact that the firm gets the full surplus from the acquihire.\label{fn:timing}}

We expect our model to be more applicable to situations in which the startup is still relatively young and not widely recognized, making the assumption that firms become aware of it sequentially more realistic. Indeed, as evidenced in the Federal Trade Commission's report on non-Hart-Scott-Rodino acquisitions by Alphabet, Amazon.com, Apple, (then still) Facebook, and Microsoft from 2010-2019, such situations seem to be quite prevalent if not the norm \citep{FTCReport2021}. The report shows that (i) approximately 65\% of all transactions cost less than \$25m, (ii) almost half of all transactions comprised startups that were less than five years old,
(iii) more than 50\% of the transactions involved targets with fewer than ten employees. Overall, this supports the notion of small, relatively young and not widely known startups.\footnote{Should the startup already be more mature or the founder(s) actively looking for an exit, a simultaneous-move model might be more appropriate. We discuss this alternative setting and its implications on the prevalence of talent hoarding in Section \ref{sec:robustness}.} Moreover, while the match quality $\theta$ is the acquirer's private information, we view this as information that has to be learned by interacting with and getting to know the startup. Consequently, a firm that does not yet know the startup (well enough) cannot engage in costly signaling to prevent a low-type competitor of preemptively acquiring the startup.\footnote{Even if a high-type acquirer could signal its type, a classical hold-up problem may arise, as it is unclear that it could credibly commit to paying more than a low-type acquirer. Indeed, it is a priori uncertain which of the two types would have a higher willingness to pay for the startup, as can be seen from the discussion of a simultaneous-move variant in Section \ref{sec:robustness}.}

\paragraph{Example.} Our reduced-form model is consistent with many standard oligopoly models. With a specific application in mind, one could fix a demand function and derive more precise results. For example, consider a Cournot duopoly with (inverse) demand function $P(q_1, q_2)=a-bq_1-bq_2$ and  constant marginal cost of production $c$. Let a high-match acquihire reduce the acquirer's marginal cost to $c-H>0$ while a low-match acquihire reduces it to $c-L$, with $H>L$. Assuming that both firms are active after a high-match acquihire (that is, $a-c>H$), it is easy to calculate the firms' profits after the various outcomes: 
\begin{align*}
   \Pi_F = \dfrac{(a-c)^2}{9b}, \quad \quad \quad \quad   \bar{\Pi}_F^H & = \dfrac{(a-c+2H)^2}{9b}, &&\bar{\Pi}_F^L  = \dfrac{(a-c+2L)^2}{9b}, \\
    \underline{\Pi}_F^H & = \dfrac{(a-c-H)^2}{9b},   &&\underline{\Pi}_F^L = \dfrac{(a-c-L)^2}{9b}
\end{align*}
It can be shown that Assumption  \ref{ass:A1}(i) is satisfied for an interval of $\pi_E$ values (which is essentially a free parameter), while \ref{ass:A1}(ii) is always satisfied. Moreover, standard calculations give us consumer surplus for the three possible outcomes:
\begin{align*}
    CS_H  = \dfrac{(2a-2c+H)^2}{18b}, \quad \quad \quad  CS_L  = \dfrac{(2a-2c+L)^2}{18b}, \quad \quad \quad  CS_F = \dfrac{(2a-2c)^2}{18b}  
\end{align*}
It follows immediately that Assumption \ref{ass:A2} is satisfied.\hfill $\square$

\section{Talent Hoarding}\label{sec:TH}

We define talent hoarding as a situation in which a firm employs a group of workers although they could be more efficiently employed elsewhere. In our model, talent hoarding occurs whenever a low-match firm acquires and integrates the startup because the employees of the startup would generate higher profits if it remained operational. Moreover, if the acquiring firm's competitor turns out to be a high match with the startup, the forgone efficiency is  even greater.

\begin{proposition}[Talent hoarding] \label{prop:AHonly}
	Under Assumption \ref{ass:A1},  firm 1's behavior in any perfect Bayesian equilibrium (PBE) is uniquely specified. Namely, if firm 1 is a high match with the startup, it will pursue an acquihire; if it is a low match, it will pursue an acquihire if and only if
\begin{align} \label{eq:lambda}
	\lambda \geq \lambda_A\equiv\frac{\pi_E+\Pi_F-\bar{\Pi}_F^L}{\Pi_F-\underline{\Pi}_F^H}.
\end{align}
\end{proposition}

\vspace{2mm}

\begin{proof}
	Suppose firm 1 has not done an acquihire. It follows from Assumption \ref{ass:A1} that, for any belief, firm 2 does an acquihire if and only if it is a high type. Moving to stage 1, firm 1's belief is given by the prior. A high-match firm 1 will always do an acquihire by Assumption \ref{ass:A1}. A low-match firm 1  will do so whenever $\bar{\Pi}_F^L - \pi_E \geq \lambda \underline{\Pi}_F^H + (1-\lambda)\Pi_F$ or, equivalently,
	\begin{align} \notag
		\lambda>\lambda_A\equiv\frac{\pi_E+\Pi_F-\bar{\Pi}_F^L}{\Pi_F-\underline{\Pi}_F^H}.
	\end{align}
	Note that $\pi_E$ is firm 1's bid for the startup, leaving the entrepreneur just indifferent between accepting and not accepting.
\end{proof}

\vspace{3mm}

The result in Proposition \ref{prop:AHonly} shows that talent hoarding may occur when (i)  $\bar{\Pi}_F^L  - \underline{\Pi}_F^H > \pi_E$ so that $\lambda_A<1$, and (ii) the probability of a high match is sufficiently high. The condition $\bar{\Pi}_F^L  - \underline{\Pi}_F^H > \pi_E$ guarantees that the gain for a low-match firm from  an acquihire when facing a high-match competitor is bigger than the cost of the acquihire. However, since a low-match firm makes a negative profit from the acquihire per se, it will only proceed when facing a high-match competitor is likely enough.\footnote{In classical labor models where firm-employee match value matters \citep[e.g.,][]{Jovanovic1979}, firms only care about their own match value. In our model, because of  oligopolistic competition, the potential match value between the competitor and the worker is driving the results.} Effectively, a low-match firm 1 is willing to overpay when making the acquihire to prevent the potential emergence of a highly competitive firm 2. While a low-match firm 1 does reap some efficiency gain from the acquihire, it is the threat of a more competitive firm 2 that motivates the acquihire. Thus, talent hoarding is more likely if the price of the acquisition $\pi_E$ is low and the probability of a high-match competitor $\lambda$ is high.\footnote{One may wonder why the low-match firm does not keep the startup operational or what would happen if the startup was \emph{also} valuable because of its technology. We discuss these issues in Section \ref{sec:extensions}. Briefly, we argue that operating the startup as a subsidiary might be less profitable than integration due to moral-hazard issues. Indeed, after the acquisition of Zenly by Snap (see footnote \ref{fn:zenly}), the startup was kept essentially independent. One can interpret the owner's decision to shut down the startup as recognizing that Zenly's autonomy encouraged risky expansion over profitability. Further, we also show that when startups own valuable technology, the firms  have no incentive to hoard technology but  still hoard talent.}

The discussion so far has focused exclusively on the firms. We now turn to the implications of talent hoarding for consumers. Following an acquihire, the consumer surplus generated by the startup vanishes. Thus, whether consumers benefit from the acquihire, and hence what  the appropriate response of the regulators is, depends on the change in consumer surplus created by the competition between the firms following the acquihire.

By Assumption \ref{ass:A2}, acquihires always increase consumer surplus in the two-firm market ($CS_H > CS_L > CS_F$) but lead to the loss of $CS_E$  in the startup market. If $CS_E$ is very low, so that $CS_H > CS_L > CS_F + CS_E$, then both the low-match and high-match acquihires increase consumer surplus and the policy makers should always allow acquisitions.\footnote{Observe that when $CS_E=0$ (e.g., because the startup is not viable) an acquihire always increases consumer surplus.} Similarly, if $CS_E$ is very high so that both the low-match and high-match acquihires lower consumer surplus  ($CS_F + CS_E > CS_H > CS_L $), then the policy makers should prohibit all acquisitions.

A more subtle case appears if $CS_E$ is intermediate, so that $CS_H > CS_F + CS_E > CS_L $.  Now, prohibiting a high-match acquihire would decrease consumer surplus, while prohibiting a low-match one would increase it. Hence, a policy maker unable to discern low- from high-match acquihires faces a trade-off. Our next result characterizes the optimal policy in all the cases discussed. Define 
\begin{align}
	\lambda_{CS} \equiv \dfrac{ CS_F + CS_E - CS_L}{CS_H - CS_L}.
\end{align}
How the values of this cutoff and the one defined in \eqref{eq:lambda} relate to each other are crucial for the effect of acquihires on consumer surplus.

\begin{proposition}[Effect of acquihires on consumer surplus]\label{prop:CS} $ $\begin{enumerate}[label=(\roman*)]
		\item If $CS_F + CS_E > CS_H > CS_L$, then all acquisitions reduce consumer surplus.
		\item If $CS_H > CS_L > CS_F + CS_E$, then all acquisitions increase consumer surplus.
		\item Suppose that $CS_H > CS_F + CS_E > CS_L$. Acquihires reduce consumer surplus in expectation if and only if $\lambda \in [ \lambda_{A},\lambda_{CS})$.
	\end{enumerate}
	
\end{proposition}

\vspace{2mm}

\begin{proof}
	Cases {(i)} and {(ii)} are straightforward. We demonstrate {(iii)}. When $\lambda < \lambda_A$, by Proposition \ref{prop:AHonly} only high-match firms engage in an acquihire. Since $CS_H > CS_F + CS_E$, any acquihire in this case increases consumer surplus. When $\lambda \geq \lambda_A$, both low-match and high-match firm 1 engage in an acquihire and the expected consumer surplus is 
	$
		\lambda CS_H + (1-\lambda)CS_L.
	$
	Since the expected consumer surplus when acquihires are prohibited is $CS_F + CS_E$, acquihires reduce consumer surplus if and only if $\lambda \geq \lambda_A$ and  
	\begin{align*}
		\lambda CS_H + (1-\lambda)CS_L  < CS_F + CS_E \Longleftrightarrow
		\lambda  < \dfrac{ CS_F + CS_E - CS_L}{CS_H - CS_L} \equiv \lambda_{CS}.
	\end{align*}
	Thus, acquihires reduce consumer surplus if and only if $\lambda_A \leq \lambda < \lambda_{CS}$. 
\end{proof}

\vspace{3mm}

The intuition for Proposition \ref{prop:CS} (iii) is that when $CS_H > CS_F + CS_E > CS_L$, acquihires are harmful only when there is talent hoarding (i.e., both low- and high-match firms engage in an acquihire, requiring $\lambda \geq \lambda_A$) and the probability of a high match is sufficiently low (requiring $\lambda < \lambda_{CS}$). Hence, acquihires lower expected consumer surplus only for intermediate $\lambda$, that is, when $\lambda \in [ \lambda_{A},\lambda_{CS})$. 

Figure \ref{fig:CS} illustrates Proposition \ref{prop:CS}(iii) using our Cournot example with $\pi_E=0.9$. The two panels only differ in the level of consumer surplus generated by the startup. We set $CS_E=0.4$ for the left panel and $CS_E=0.5$ for the right panel. In both panels, as $\lambda$ grows from $0$ to $\lambda_A$, the consumer surplus when acquisitions are allowed increases (the solid line). For these parameters, firms endogenously only engage in high-match acquihires, which benefit consumers. At $\lambda_A$, low-match firms start talent hoarding, causing a discontinuous drop in consumer surplus visible on both panels. However, in the left panel, $\lambda_{CS} < \lambda_{A}$, so the drop at $\lambda_A$ is not sufficient to lower the consumer surplus below the level achieved when acquisitions are prohibited (the dash-dotted line). In the right panel, $\lambda_{A} < \lambda_{CS}$, so that for all  $\lambda \in [ \lambda_{A},\lambda_{CS})$ the average consumer surplus is lower when acquihires are allowed than when they are not. As $\lambda$ increases beyond $\lambda_{CS}$, high-match acquihires become so likely that the overall consumer surplus is again higher than when acquisitions are prohibited. Thus, when $CS_H > CS_F + CS_E > CS_L$, allowing acquihires only lowers consumer surplus for \emph{intermediate} values of $\lambda$. Finally, the dashed line represents the consumer surplus that could be achieved if the regulators could differentiate between the low-match and high-match acquihires. In that case, allowing only high-match acquihires always increases consumer surplus. Of course, if the regulators can only imperfectly detect match types, that  lowers the expected consumer surplus below the dashed line.

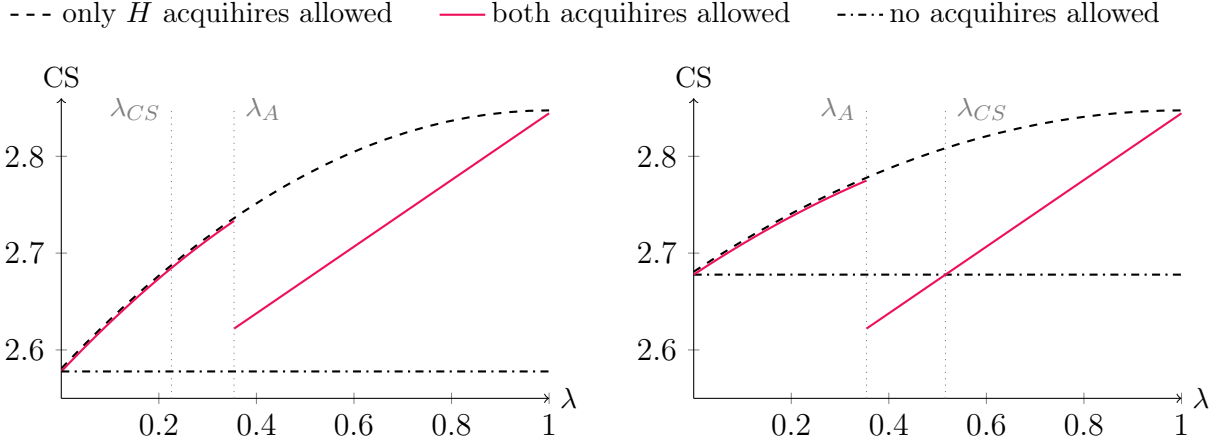
\begin{figure}[t!]
        \centering
    \begin{minipage}{0.48\textwidth}
        \centering
        \begin{tikzpicture}
            \begin{axis}[
                    width=8cm,
                    domain=0:1,
                    xmin=0, xmax=1,
                    ymin=2.55, ymax=2.86,
                    xlabel={$\lambda$},
                    ylabel={CS},
                    axis lines=middle,
                    axis line style={->},
                    xlabel style={
                    at={(axis description cs:1,0)}, anchor=west, font=\small},
                    ylabel style={at={(axis description cs:0,1)}, anchor=south, font=\small},
                    legend columns = -1,
                    legend style={
                at={(0.7,1.2)},
                anchor=south,
                draw=none,
                font=\small
                },
                    unit vector ratio=1 2
                ]
                \addplot[black, thick, dashed]{0.003 + 2.84444*(1 - (1 - x)^2) + 2.57778*(1 - x)^2}; % CSProb
                \addlegendentry[align=left]{only $H$ acquihires allowed \hspace{4mm}}
                \addplot[OrangeRed, thick, domain=0:0.354167, forget plot]{2.57778*(1 - x)^2 + 2.84444*(1 - (1 - x)^2)}; % CSAllowed before 0.354167
                \addlegendentry[align=left]{both acquihires allowed}
                \addplot[OrangeRed, thick,  domain=0.354167:1]{2.5*(1 - x) + 2.84444*x}; % CSAllowed after 0.354167
                \addplot[black, thick, dash dot]{2.57778}; % CSban
                \addplot[color=gray, dotted] coordinates {(0.354167,2.55) (0.354167,2.85)} node[pos=1, right,font=\small] {$\lambda_A$};
                \addplot[color=gray, dotted] coordinates {(0.225806,2.55) (0.225806,2.85)} node[pos=1, left,font=\small] {$\lambda_{CS}$};
            \end{axis}
        \end{tikzpicture}
        \end{minipage}\hfill
            \begin{minipage}{0.48\textwidth}
        \centering
        \begin{tikzpicture}
            \begin{axis}[
                    width=8cm,
                    domain=0:1,
                    xmin=0, xmax=1,
                    ymin=2.55, ymax=2.86,
                    xlabel={$\lambda$},
                    ylabel={CS},
                    axis lines=middle,
                    axis line style={->},
                    xlabel style={
                    at={(axis description cs:1,0)}, anchor=west, font=\small},
                    ylabel style={at={(axis description cs:0,1)}, anchor=south, font=\small},
                    legend columns=2,
                    legend style={
                at={(0.63,1.2)},
                anchor=south,
                draw=none,
                font=\small},
                    unit vector ratio=1 2
                ]
                \addplot[black, thick, dash dot]{2.67778}; % CSban
                \addlegendentry[align=left]{no acquihires allowed}
                \addplot[black, thick, dashed]{0.003 + 2.84444*(1 - (1 - x)^2) + 2.67778*(1 - x)^2}; % CSProb
                \addplot[OrangeRed, thick, domain=0:0.354167, forget plot]{2.67778*(1 - x)^2 + 2.84444*(1 - (1 - x)^2)}; % CSAllowed before 0.354167, not included in the legend
                \addplot[OrangeRed, thick,  domain=0.354167:1]{2.5*(1 - x) + 2.84444*x}; % CSAllowed after 0.354167
                \addplot[color=gray, dotted] coordinates {(0.354167,2.55) (0.354167,2.85)} node[pos=1, left,font=\small] {$\lambda_A$};
                \addplot[color=gray, dotted] coordinates {(0.516129,2.55) (0.516129,2.85)} node[pos=1, right,font=\small] {$\lambda_{CS}$};
            \end{axis}
        \end{tikzpicture}
    \end{minipage}
    \caption{Consumer surplus in the Cournot example when Proposition \ref{prop:CS}(iii) applies. In the left panel $\lambda_{CS}<\lambda_A$ so that all acquihires increase consumer surplus. In the right panel $\lambda_A<\lambda_{CS}$ so that acquihires decrease consumer surplus if and only if $\lambda \in [\lambda_A, \lambda_{CS}]$.}
    \label{fig:CS}
\end{figure}

Finally, the regulator may be interested in the effect of talent hoarding on total surplus, that is, the sum of firms' profits and consumer surplus. It follows directly from our discussion of Propositions \ref{prop:AHonly} and \ref{prop:CS} that talent hoarding reduces total surplus unambiguously when the transaction reduces consumer surplus. Otherwise, the impact depends on the structure one imposes on our reduced-form model. In the Cournot example, total surplus with a low-type acquihire is always lower than with a high-type acquihire, while the total surplus without an acquihire hinges on the startup's profits and consumer surplus.

\section{Robustness and Limits to Regulation} \label{sec:robustness}

Propositions \ref{prop:AHonly} and \ref{prop:CS} provide conditions -- within our baseline framework -- under which talent hoarding is likely to occur and when acquihires are harmful for consumer surplus, respectively. Taken together, these results also shed light on the welfare impact of talent hoarding. However, it is essential to assess the extent to which these findings hold beyond our baseline framework. To this end, we explore several variations of our benchmark model. Through these extensions, we clarify the robustness and limitations of our results regarding the prevalence of talent hoarding and discuss additional factors that policymakers should consider when regulating acquihires. The formal analysis underlying this section is presented in Appendix \ref{sec:app_extensions}.

\paragraph{Simultaneous Moves.} When both firms are aware of the startup, as would be the case when the entrepreneur is already established and well-known, modeling a simultaneous-move game might be more appropriate. This scenario effectively corresponds to an auction with externalities, as described by \citet{JehielMoldovanuStacchetti1996}. Although not all of their results are directly applicable here, some insights are relevant. To begin with, consider the case of complete information and assume that one firm has a high and the other a low type. Then, the entrepreneur optimally sells the startup to the low-type firm if and only if 
\begin{align} \label{eq:simultaneous}
    \bar{\Pi}_F^L + \underline{\Pi}_F^L \geq \bar{\Pi}_F^H + \underline{\Pi}_F^H. 
\end{align}
Hence, the low-type firm outbids the high-type firm whenever the sum of their profits is greater after a low-type acquisition than after a high-type acquisition.\footnote{This condition is unlikely to hold in standard models with homogeneous goods, and in particular it is not satisfied in our Cournot example. However, it can hold in models with differentiated goods where each firm has a large base of ``loyal'' consumers.} While this outcome is efficient from an industry-wide perspective, talent hoarding still arises as the workers would be more efficiently employed at the startup.  Conversely, when the inequality in \eqref{eq:simultaneous} is reversed, the high-type firm outbids the low-type firm, thereby eliminating talent hoarding. Extending this analysis to an incomplete-information setting does not alter the qualitative insights. 

Overall, in simultaneous-move settings, talent hoarding remains a possibility but is likely to be less prevalent than in our benchmark model. Consequently, the associated welfare concerns are attenuated, suggesting that talent hoarding may not be a strong justification for blocking startup acquisitions in these situations.

\paragraph{Uncertain Order of Moves.} As previously noted, talent hoarding is more likely to occur when firms can identify high-impact individuals before their competitors. However, when a firm becomes aware of a startup with such employees, it may not know whether its competitors are also aware. To capture this uncertainty, we extend the model by introducing an unknown order of moves, where the sequence is privately determined by nature with uniform probabilities. In this setting, a symmetric equilibrium exists in which all firms undertake an acquihire regardless of their type if
\begin{align} 
		\lambda \geq \lambda_{A,U}\equiv\frac{\pi_E+\underline{\Pi}^L_F-\bar{\Pi}_F^L}{\underline{\Pi}^L_F-\underline{\Pi}_F^H}.
\end{align}
This condition is more stringent than in the benchmark model, but talent hoarding may still arise. Nonetheless, the likelihood of adverse welfare effects diminishes, as the uncertainty surrounding the order of moves reduces the prevalence of talent hoarding.

\paragraph{Emerging Competitor.} We now consider a scenario where the two incumbents do not co-exist. That is, only one incumbent is initially active, but the other firm may yet enter.\footnote{This setting with a potentially emerging competitor is reminiscent of Segal and Whinston (2000). However, allowing for ex-post renegotiation with an emerging competitor, as in their paper, has an ambiguous effect on the incentives to hoard talent. If a low-type firm could sell the startup to an emerging high-type firm, the incentives to acquire the startup may increase, as the loss incurred through the acquisition can be recouped when selling to the high-type competitor. At the same time, the possibility that the competitor does not even emerge reduces the incentives to acquihire the startup in the first place, leaving the total effect ambiguous.} Specifically, the entry of the second firm only happens with probability $q\in [0,1]$. If the second firm does not enter, and there is no acquisition, the first firm makes a monopoly profit $\Pi_M>\Pi_F$. Acquihiring the startup yields a match-dependent monopoly profit of $\bar{\Pi}_M^H$ or $\bar{\Pi}_M^L$, where---mirroring Assumption \ref{ass:A1}---we assume $\bar{\Pi}^H_M>\Pi_M+ \pi_E>\bar{\Pi}^L_M$. If the second firm does enter, the same profits as in the benchmark accrue (with or without an acquihire). In the case of an emerging competitor, the condition for talent hoarding to emerge is given by 
\begin{align}
    \lambda \geq \lambda_{A, M}\equiv\frac{\pi_E - q(\bar{\Pi}_F^L - \Pi_F)- (1-q)(\bar{\Pi}_M^L - \Pi_M)}{q(\Pi_F - \underline{\Pi}_F^H)}.
\end{align}
As $q$ increases, this condition approaches condition \eqref{eq:lambda} from the benchmark model. However, as $q$ decreases -- and thus the threat of an entrant who might acquihire the startup diminishes -- the condition becomes less likely to hold. In such a situation, adverse welfare effects are less likely than in the benchmark, as talent hoarding is less likely to occur and consumer surplus is less likely to be reduced. Finally, as the entry of the competitor is uncertain to begin with, the regulator ought to be cautious when considering the ban of startup acquisitions. 

\paragraph{Surplus Sharing} In our baseline model, the acquiring firm has all the bargaining power and thus manages to acquire the startup at a price that leaves the entrepreneur with zero surplus. In this extension, we show that talent hoarding can persist even when surplus is shared, provided the entrepreneur's bargaining power is not too large. Formally, we assume that in case of an acquihire the entrepreneur receives a share $\xi \in [0,1]$ of the surplus. Then, the condition for talent hoarding becomes 
\begin{align}
    \lambda \geq \lambda_{A,S} \equiv\frac{\pi_E+\Pi_F-\bar{\Pi}_F^L}{\Pi_F-\underline{\Pi}_F^H - \xi(\bar{\Pi}_F^H - \Pi_F-\pi_E)}.
\end{align}
Evidently, as $\xi$ and thus the entrepreneur's bargaining power increases, the condition is less likely to be satisfied. In particular, if 
\begin{align*}
	\xi > \frac{\bar{\Pi}_F^L - \underline{\Pi}_F^H-\pi_E}{\bar{\Pi}_F^H - \Pi_F-\pi_E}.
\end{align*}
the acquisition of the startup becomes prohibitively expensive for the firm and talent hoarding does not emerge in equilibrium. Once again, adverse welfare impacts are less likely. The result suggests that regulators could try to use estimates of the bargaining power in startup acquisitions to gauge whether an acquisition should be blocked or not.

\paragraph{Dominant Firm.} Instead of symmetric firms, the market could be characterized by a dominant firm and a challenger. We show that the incentives to hoard talent also emerge in this asymmetric setting. Specifically, suppose that the dominant firm both stands to lose more if a startup is acquired by the challenger (maybe because the dominant firm has a larger market share, so it has more to lose) and stands to gain more by acquihiring the startup itself (e.g., because its larger market share allows it to deploy improvements more rapidly to a broader consumer base). Moreover, assuming that a dominant firm makes more profit than symmetric firms (i.e., $\Pi_D \geq \Pi_F$) and is more likely to identify startups than its competitor, then talent hoarding is more likely to emerge than in the benchmark model, as the relevant condition now reads 
		\begin{align}			\lambda>\lambda_{A,D}\equiv\frac{\pi_E+\Pi_D-\bar{\Pi}_D^L}{\Pi_D-\underline{\Pi}_D^H},
		\end{align}
which is less stringent than condition \eqref{eq:lambda}. This suggests that acquisitions by dominant firms may deserve more scrutiny by the regulator, as they are more likely to lead to talent hoarding and negative effects on consumer welfare.  However, an outright ban on such acquisitions may not be justified, as acquisitions by challenger firms are less likely to produce these adverse outcomes. 

\paragraph{Multiple Firms.} Finally, we examine markets with multiple firms that may sequentially attempt to acquihire the startup. Focusing on the Cournot-Oligopoly setting, we show that in the limit, as the number of firms grows large, talent hoarding ceases to take place. Intuitively, the profit at risk from a competitor's acquihire becomes smaller as the number of competitors increases, reducing the incentive to engage in costly talent hoarding. However, the effect of the increase in the number of competitors on the incentive to hoard talent is not necessarily monotonic. The reason for the non-monotonicity is that an increase in the number of competitors, in addition to decreasing profit at risk, also increases the probability that a high-match competitor will acquire the startup. For a small number of firms, it is unclear which effect dominates. Indeed, we show in a parametric example that talent hoarding may be more prevalent with three firms than with two. Even if talent hoarding may arise with more than two firms, the concerns about adverse welfare effects should be less pronounced when markets are more competitive.

\section{Hiring and Layoffs}

We now discuss the implications of talent hoarding on the hiring and layoffs of acquihired employees. To do so, we expand our baseline model by introducing a second period and allowing for economic downturns between the two periods.

The first period of this expanded model is identical to the baseline model in Section \ref{sec:model}: the firms' match types are private information, and firm 1 has the opportunity to pursue an acquihire before firm 2. Between the periods, the economy suffers a downturn with probability $\delta\in (0,1)$. If a downturn materializes -- an event that is publicly observable -- the firms may be hit by adverse shocks (see details below). In period 2, the entrepreneur has the option of creating a new startup, once more leading to an outside option of $\pi_E$ for her.\footnote{There is empirical evidence that acquihired employees who leave the acquirer are likely to join a new startup \citep{Ng+Stuart2021, kim2022startup, Kim2020startup}.} If the entrepreneur was employed by a firm in period 1, that firm must decide whether to continue the relationship (at the cost of $\pi_E$) or lay off the entrepreneur, who might then be hired by the other firm. If there was no acquihire in period 1, it is again firm 1 that moves first in period 2.

The adverse shocks come with a commonly known joint probability distribution over ``downgrades'' for firms, though each firm only observes the realization of its own shock. More specifically, firm $i$ is hit by a shock $S_i \in \{D, N\}$, where it is either downgraded from high to low match (if possible) or not affected by the shock. Thus, if a low-match firm is hit by a downgrade, it stays a low-match firm, while a high-match firm turns into a low-match firm. If a firm is not affected by the shock, its match quality remains the same.\footnote{Of course, an economy-wide adverse shock may affect more parameters than just the profitability of the high-match firm (for example, it may impact the profitability of startups or firms even in the absence of any acquisitions). We are focusing on the impact on the high-match firms as they will presumably seek to expand, the profitability of which is particularly likely to be impacted by a negative shock to the economy.} Let $(S_1, S_2) \in \{ D, N\}^2$ be the profile of shocks hitting the firms, which follows the distribution 
\begin{align*}
	\Pr(D, D) &= r \gamma(1-\gamma)  +\gamma^2, &&\Pr(D, N) = (1-r) \gamma(1-\gamma), \\
	\Pr(N, D) &= (1-r) \gamma(1-\gamma), 
	&&\Pr(N, N) = r \gamma(1-\gamma)  +(1-\gamma)^2,
\end{align*}
where $\gamma\in(0,1)$ is the probability that a firm  will be downgraded and $r\in [0,1]$ measures the positive correlation between the firms' shocks. In particular, for $r=0$ the shocks are independent, and for $r=1$ they are perfectly positively correlated. 

To make a comparative statement, we need a benchmark relative to which we can compare the hiring and layoffs in our model with talent hoarding. To do so, consider the case in which $\Pi_F = \underline{\Pi}_F^H = \underline{\Pi}_F^L$ so that an acquihire by firm $i$ does not affect firm $j$'s profits. Thus, there are no incentives to hoard talent in this benchmark. 

Before we state the formal result, consider the following intuition. If the entrepreneur was hired by a high-match firm in period 1 and this firm is not affected by the economic downturn (i.e., remains high-match), the firm will continue to employ the entrepreneur. Therefore, no layoff is observed in that case.\footnote{In the video game industry, \cite{loh2019disruption} document that when the skills of the employees and the needs of the acquirer match well, the employees are more likely to stay with the acquirer.} In contrast, if the entrepreneur was hired by a low-match firm in period 1 or a downgraded high-match firm, several period 2 outcomes can arise. Talent-hoarding motives may induce the continued employment of the entrepreneur if the competitor is believed to have a high match value with a high enough probability. Thus, we do not observe any layoff of the entrepreneur. Otherwise, we may observe a layoff of the entrepreneur, who is subsequently hired by a high-match competitor. Hence, while we observe a separation, the entrepreneur remains employed in the industry. Finally, the entrepreneur may be laid off and not hired by the competitor, resulting in the entrepreneur leaving the industry. Note that these distinct period 2 outcomes, in turn, affect the behavior of low-match firms in period 1, changing the acquihire threshold. Solving the game fully and deriving the probabilities of period 1 hiring as well as observing a layoff in period 2, we obtain the following result, proved in the appendix. 

\begin{proposition}[Effect on employment outcomes] \label{prop:labor}
	The presence of talent-hoarding motives always leads to weakly more hiring than in the benchmark. Additionally, provided that  $\min\left\{\frac{\lambda_A}{\lambda}, \frac{1-\lambda}{\lambda}\right\} > (1-r)(1-\gamma)$, talent hoarding also leads to weakly more layoffs than in the benchmark.
\end{proposition}

The increase in hiring follows immediately because in the presence of talent hoarding, not only high-match firms but also low-match firms may pursue an acquihire. The increase in layoffs is more subtle. 
Essentially, when either the correlation between firms' adverse shocks or the (marginal) probability of suffering a downgrade $\gamma$ is sufficiently high, talent hoarding raises layoffs. In the case in which $r$ is high, this is because firm 1's shock is informative of firm 2's shock because of the correlation, hence allowing firm  1, whenever it draws a negative shock, to forgo the costly talent hoarding in the second period. 
Similarly, when $\gamma$ is sufficiently high, firm 1 can be fairly confident that whenever a downturn occurs, the competitor will be downgraded, once more prompting the firm to forgo talent hoarding. 
Further, when $\gamma$ or $r$ is sufficiently high, firm 1 is often right in laying off the entrepreneur, as firm 2 will indeed have a low match, which in turn will lead to the entrepreneur exiting the industry. Finally, all statements in the proposition are strict whenever $\lambda$ is sufficiently high so that any talent hoarding at all takes place.

The period towards the end of the COVID pandemic serves as a good example for the effect we aim to capture. Analysts have argued that rapid hiring by tech companies during the pandemic was fueled by a massive positive demand shock tied to work-from-home trends and by historically low interest rates that made borrowing inexpensive. As an adverse shock in the form of employees returning to the office and increased interest rates hit big tech, companies like Google were left with over-expanded workforces, prompting widespread layoff \citep{NYT,fDi}. Our model suggests that the size of these layoffs may have been magnified by talent hoarding. Indeed, John Gruber, a tech commentator,  wrote the following: ``There are numerous reasons the tech industry wound up at this layoffpalooza, but I think the main reason is that the biggest companies got caught up in a game where they tried to hire everyone, whether they needed them or not, to keep talent away from competitors and keep talent away from small upstarts (or from founding their own small upstarts). These big companies were just hiring to hire, and now the jig is up.'' \citep{Gruber_layoffs} Sam Lessin, the founder of Drop.io mentioned in the introduction, expressed a similar sentiment in an interview \citep{Lessin_layoffs}. Overall, while the main economic forces that our model identifies are unlikely to capture the complete story behind these layoffs, they align with recent developments and are supported by industry insiders' observations.

\section{Other Extensions}\label{sec:extensions}

We now extend our baseline model in two directions, both of which lead to more talent hoarding. First, we allow for the situation in which the startup is valuable to the firms not only because of its employees but also because of its technology. In doing so, we emphasize the difference between labor and other inputs in that firms do not acquire property rights over labor. Thus, in contrast to the labor, the technology can be sold (or licensed) after an acquisition. Second, we consider partial acquisitions, or investments, into startups. In this case, the acquirer does not acquire the startup in its entirety but only a fraction of it and in turn receives a fraction of the startup's profits as well as some decision power. We show that when the decision power suffices to block acquisitions by competitors, such investments constitute an attractive alternative to full acquihires. As it turns out, the possibility of investments is more likely to lead to inefficient market outcomes, albeit at a lower degree of inefficiency.

\subsection{People and Technology}\label{sec:technology}

Consider a situation where the total value of the startup consists of the people who work for the startup and the technology owned by the startup. The fundamental difference between the employees and the technology, from the acquirer's point of view, is that technology can be sold (or licensed), while the people cannot. In our model, this implies that the acquirer can resell the startup's technology to the competitor, whenever such a sale increases joint profits. Suppose that the share of the value of the startup generated by the technology is $\tau \in [0,1]$, while the share generated by employees is $(1-\tau)$. Moreover, for simplicity, assume that acquiring just the technology (or just the employees) generates $\tau$ (or $1-\tau$) of the impact that acquiring the entire startup would have. Just as before, a firm can have a high match value with the startup with probability $\lambda$, where we assume for simplicity that the match value of the startup to a firm applies to both the people and the technology identically. The match value of the firm is private information at the beginning of the game. The timing of the game in stage 1 is:
\begin{enumerate}
	\item Firm 1 observes the match quality with the startup and makes an acquisition of the startup at price $p$ or does nothing.
	\item The startup accepts or rejects the bid. 
	\item If the bid is rejected, the game proceeds to stage 2. If accepted, firm 1 can sell the startup's technology at the price $q$ to firm 2. 
\end{enumerate}
Stage 2 is like stage 1 but the roles of firms 1 and 2 are reversed. To accommodate the possibility of selling the startup's technology, we slightly adapt the notation from the baseline model. Suppose firm 1 with match $\theta_1$ did an acquisition at price $p$. Absent any sale of the startup's technology, profits read 
\begin{align*}
	\text{Firm} ~1:&\; \Pi_F + \bar{\pi}_F^{\theta_1}-p\\
	\text{Firm} ~2:& \;\Pi_F-  \underline{\pi}_F^{\theta_1}\\
	\text{Startup} : &\; p,
\end{align*}
which coincides with profits in the baseline model (although the notation is different).\footnote{For instance, the payoff of an acquiring firm with match type $\theta_1$ in the main text is $\bar{\Pi}_F^{\theta_1}$ while it now reads $\Pi_F + \bar{\pi}_F^{\theta_1}$.} If the technology part is sold at price $q$, the profits read
\begin{align*}
	\text{Firm} ~1:& \;\Pi_F + (1-\tau)\bar{\pi}_F^{\theta_1} - \tau\underline{\pi}_F^{\theta_2} -p+q\\
	\text{Firm} ~2:& \;\Pi_F-  (1-\tau)\underline{\pi}_F^{\theta_1}+ \tau \bar{\pi}_F^{\theta_2}-q\\
	\text{Startup} : & \;p.
\end{align*}
Thus, the people who have joined firm 1 from the startup increase firm 1's profits and decrease firm 2's profits, respectively. Conversely, the startup's technology increases firm 2's and decreases firm 1's profits, respectively. 

To simplify the model and the exposition, we do not explicitly model the bargaining process between the two firms. Instead, we assume that the two firms meet at the bargaining table, their types are revealed and the resulting surplus from selling the technology is shared equally. The surplus resulting from the sale of the technology is then given by
\begin{align*}
	\Pi_F &+ (1-\tau)\bar{\pi}_F^{\theta_1} - \tau\underline{\pi}_F^{\theta_2} -p+\Pi_F-  (1-\tau)\underline{\pi}_F^{\theta_1}+ \tau \bar{\pi}_F^{\theta_2}-\left(\Pi_F + \bar{\pi}_F^{\theta_1}-p+ \Pi_F-  \underline{\pi}_F^{\theta_1}\right)\\
	=& \tau\left(\bar{\pi}_F^{\theta_2}+ \underline{\pi}_F^{\theta_1}-\bar{\pi}_F^{\theta_1}-\underline{\pi}_F^{\theta_2}\right).
\end{align*}
The case we are interested in, is when a low-match firm 1 sells technology to a high-match firm 2. Then, the surplus reads $\tau\left(\bar{\pi}_F^{H}+ \underline{\pi}_F^{L}-\bar{\pi}_F^{L}-\underline{\pi}_F^{H}\right)$. We now make two assumptions on profits. The first simply restates Assumption 1 into this section's notation. The second ensures that the surplus resulting from a technology sale from a low-match to a high-match firm is positive.

\begin{assumption} \label{ass:A3}\quad
	$
		\bar{\pi}_F^{H} > \pi_E > \bar{\pi}_F^{L} ~\text{and}~
	\underline{\pi}_F^{H} > \underline{\pi}_F^{L}\geq 0.
$
\end{assumption}

\begin{assumption}\label{ass:A4}\quad
$
    \bar{\pi}_F^{H}+ \underline{\pi}_F^{L}-\bar{\pi}_F^{L}-\underline{\pi}_F^{H}>0.
$		
\end{assumption}

Finally, to break ties, we assume that firms of the same match type do not trade the startup's technology. We obtain the following result

\begin{proposition} \label{prop:people_tech}
	Under Assumptions \ref{ass:A3} and \ref{ass:A4},  firm 1's behavior in any PBE is uniquely specified. Namely, if firm 1 has a high match, it will make an acquisition and not sell the technology; if it has a low match, it will make an acquisition and sell the startup's technology to a high-match (but not a low-match) firm 2 if and only if 
	\begin{align*}
		\lambda \geq \lambda_A(\tau) \equiv \frac{\pi_E - \bar{\pi}_F^L}{\underline{\pi}_F^H + \frac{\tau}{2}\left(\bar{\pi}_F^{H}+ \underline{\pi}_F^{L}-\bar{\pi}_F^{L}-\underline{\pi}_F^{H}\right)}
	\end{align*}
 and do nothing otherwise.
\end{proposition}

One can verify that for $\tau=0$ the above condition  reduces to the condition \eqref{eq:lambda}  in the baseline model. As $\tau$ increases, the threshold $\lambda_A(\tau)$ decreases, i.e., the acquisition happens for a larger set of parameters. Hence, a low-match acquirer indeed has an incentive to sell the technology to a high-match competitor because the price compensates her for any decrease in profit due to facing a more efficient competitor. Thus firms do not have an incentive to hoard technology, while the incentive to hoard talent remains. Interestingly, talent hoarding now occurs for a strictly larger set of parameters, as the option to resell technology effectively subsidizes talent hoarding.

\subsection{Partial Acquisitions}\label{sec:investments}

This section extends the baseline model to incorporate the possibility of partial investments. That is, a firm may acquire a stake (possibly a minority stake)  in the startup without integrating it. To do so, we need to specify the payoffs resulting from such partial acquisitions as well as the rights that come with partial ownership. 

Formally, the ability to make partial acquisitions means that firms can, in addition to full acquihires, try to acquire a share $s\in (0,1]$ of the startup and continue to operate it as a stand-alone entity. In contrast to the baseline model, we allow for upfront and deferred payments $(p, d)$ as is standard in such transactions.\footnote{In the baseline model, allowing for for upfront and deferred payments would not change anything.} This allows the acquiring firm to distinguish at least partially between the investor (who only gets a part of the upfront payment) and the entrepreneur (who can also receive deferred payments). If a firm acquires a share $s$ with bid $(p,d)$, its payoff reads $\Pi_F + s\pi_E(s) - p -d$, while the entrepreneur's payoff is $(1-s)\pi_E(s) + w(s) + p + d$. Here, $\pi_E(s)$ captures the startup's profit net of potential wages paid to the entrepreneur as a function of the size of the external ownership. These profits accrue to the firm and the entrepreneur proportionate to their stake in the startup. Correspondingly, $w(s)$ constitutes the entrepreneur's wage (net of effort costs) for different degrees of outside ownership. In particular, $\pi_E = \pi_E(0) + w(0)$. That is, when the entrepreneur owns the entire startup and thus obtains all of its profit and ``pays herself'' a net-of-effort wage, her payoff coincides with the initial payoff in the baseline model. The other firm's payoff is unchanged at $\Pi_F$. When a firm acquires a stake in the startup it receives a share of profits, while the dilution of ownership gives rise to moral hazard on the entrepreneur's side, which we capture in reduced form, only imposing the following assumption.

\begin{assumption} \label{ass:A5} We assume that
$\pi_E(s)+ w(s)$ is decreasing in $s$ with $\pi_E(0)+ w(0)> \pi_E(1) + w(1)$.
\end{assumption}

Assumption \ref{ass:A5} captures the moral hazard arising when the entrepreneur no longer fully owns the startup.  $\pi_E(s) + w(s)$ being decreasing reflects the reduced effort of the entrepreneur as a result of the agency problem. $\pi_E(0)+ w(0)> \pi_E(1) + w(1)$ captures that the value of a startup fully owned by the entrepreneur is strictly higher than the value when the entrepreneur is actually an employee and has no stake in the startup. 

Finally, considering control rights, we are primarily interested in the investor's ability to block, and the entrepreneur's ability to push through, an acquihire by the investor's competitor. We assume that the entrepreneur can always block an acquihire, as she could sell her shares but refuse to work for the acquiring firm. If the entrepreneur would like to be acquihired and the investor does not want the acquihire to go through, the entrepreneur can (in the spirit of typical shareholder agreements) try to ``drag along'' the investor and force the transaction. We assume that the investor has a probabilistic chance of blocking this attempt and merely impose that the probability of successfully blocking an acquihire is increasing in the investor's stake in the startup.

The timing of the game is as follows. In the first stage, firm 1 has the opportunity to make a bid to the entrepreneur to acquire a share $s_1\in (0,1]$ of the startup or make an acquihire. If the entrepreneur accepts a bid for an acquihire, the game ends. Otherwise, we move to stage 2, the ownership structure of the startup depending on whether the entrepreneur accepted firm 1's bid. In the second stage, firm 2 can make an acquihire by making a (per-share) bid to the owner(s) of the startup, where we restrict attention to bids $(p,d)$ with $p\geq \pi_E(s_1)$, so that no owner can be expropriated. If the entrepreneur does not accept the bid, the game ends. If the entrepreneur accepts, firm 1 can either also accept or try to block the transaction and succeeds in doing so with probability $\beta(s_1)$, where $\beta$ is a weakly increasing function with $\beta(0)=0$ and $\beta(1)=1$. Nature determines whether a potential blocking attempt succeeds and the game ends. 

 As a benchmark, we first consider the case where ownership of a stake does not convey any blocking rights. 

\begin{proposition} \label{prop:woblocking}
	Under Assumptions \ref{ass:A1} and \ref{ass:A5} and without blocking rights, i.e., $q(\beta)=0$ for all $\beta\in [0,1]$, firm 1's behavior in any PBE is uniquely specified. Namely, if firm 1 draws a high match type, it will make an acquihire; if it draws a low type it will make an acquihire if and only if $\lambda \geq \lambda_A$ and do nothing otherwise.
\end{proposition}

The result in Proposition \ref{prop:woblocking} shows that partial ownership of the startup is not enough to change firm 1's behavior relative to the setting with only acquihires in Proposition \ref{prop:AHonly}. Indeed, since an investment is not profitable in itself and does not prevent a high-match firm 2 from making an acquihire, firm 1 will continue to either do nothing or make an acquihire, depending on the probability of a high-match firm 2 materializing. As the next result shows, it is ownership accompanied by some measure of control over the startup, which makes investments attractive to firm 1.

\begin{proposition} \label{prop:wblocking}
	Under Assumptions \ref{ass:A1} and \ref{ass:A5} and with blocking rights, firm 1's behavior in any PBE is uniquely specified. Namely, if firm 1 draws a high type, it will make an acquihire. If it draws a low type, there is a threshold value for $\lambda$ below which it does nothing. Above the threshold, it will do an investment and, depending on the model's parameters, there may be an even higher threshold above which it does an acquihire. 
\end{proposition}

Thus, we find that an acquihire can be more profitable than buying a startup and letting it operate independently, implying that talent hoarding persists in this extension. Further, with sufficiently strong blocking rights, an investment may constitute a viable and cheaper alternative to an acquihire. Therefore, as investments are cheaper than acquihires, the possibility of partial ownership may reduce the frequency of acquihires while increasing the frequency with which some transaction takes place. Notably, the reduction in overall profits is lower in the case of investments than in the case of low-match acquihires. Therefore, the possibility of investments is more likely to lead to inefficient market outcomes, albeit at a lower degree of inefficiency. Put differently, allowing for investments increases the extensive margin of talent hoarding but partially decreases its intensive margin.

\section{Conclusion}\label{sec:conclusion}

We have presented a simple, yet general, reduced-form model of startup acquisitions. We showed that acquihires may not only reflect firms' desire to hire talented employees but also be rooted in an incentive to engage in inefficient talent hoarding, thus potentially warranting regulators' attention. Further, we showed that acquihires may decrease consumer surplus and increase job volatility of acquihired employees, thereby giving further reasons for regulatory scrutiny. 

%Finally, following the inclusion of realistic features into our baseline model such as the startup owning valuable technology that can be resold or the possibility of investing in the startup rather than acquiring it outright, talent hoarding not only persists but may even be exacerbated.

Our paper is complimentary to the literature on the acquisition of potential competitors since it shows that welfare-reducing acquisitions can occur even when the acquirer and the target are not active (even potentially) in the same market. Moreover, as we show in Section \ref{sec:technology}, hoarding of talent can be more pernicious than preemptive acquisition of technology or physical assets. This is because, in the absence of property rights, the firms cannot negotiate over the reallocation of talent. 

Still, it is not clear from our analysis that regulators always ought to intervene when incumbents acquihire startups. As the discussion surrounding Proposition \ref{prop:CS} makes clear, even within our baseline model, restrictive policy can decrease consumer welfare by blocking efficient acquisitions. Moreover, Proposition \ref{prop:CS}(iii) clarifies that improving consumer welfare can require precise information on the likelihood that the match with an incumbent is high or low, which regulators are not likely to posses. Finally, while our results in Section \ref{sec:robustness} on the order of moves and surplus sharing show that talent hoarding persists in these settings, they also show that its incidence decreases, further limiting scenarios for which regulatory intervention is desirable. In our view, the balance of benefits and harms of stricter regulatory scrutiny of startup acquisitions remains an open question. 

We close by noting that our model gives rise to several hypotheses that could be tested empirically. First, our model predicts a positive relationship between talent hoarding and job volatility of acquihired employees. Second, an acquihire by a dominant firm is more likely to be motivated by talent hoarding. Third, increasing market competitiveness can curb talent hoarding but not always monotonically. Fourth, the strength of blocking rights implied by shareholder agreements should have an impact on the relative frequency of acquihires and investments.

\newpage

\appendix

\section{Proofs} \label{app:proof}

\subsection{Proof of Proposition \ref{prop:labor}}

\setcounter{equation}{0}
\numberwithin{equation}{section}

The result is reached in three steps: solving the benchmark game without talent hoarding, with talent hoarding, and then comparing both.

\paragraph{Benchmark.} Absent incentives to hoard talent, only high-match firms do an acquihire. Thus, a layoff only takes place if the economy enters a downturn and the period-1 acquirer is hit by an adverse shock. Further, following a layoff, we observe the entrepreneur not being hired by the competitor only if the competitor was a low-match firm or if it was a high-match firm that got hit by an adverse shock. Finally, hiring takes place in period 1 unless both firms are low-match. Taken together, the probability of observing a layoff in period 2 is
\begin{align}
	l^*=\delta [\lambda+(1-\lambda)\lambda]\gamma=\delta(2\lambda-\lambda^2) \gamma
\end{align}
and the probability of observing the entrepreneur exiting the industry in period 2 is 
\begin{align}
	u^*&=\delta(2\lambda-\lambda^2-\lambda^2(1-r)(1-\gamma))\gamma.
\end{align} 

\paragraph{Talent hoarding.}

First, note that period-2 incentives coincide with those in period 1 absent an economic downturn, ruling out any layoffs. Following a downturn, three cases arise in period 2:
\begin{itemize}
	\item Suppose firm 1 acquihired in period 1. Then, firm 2 does an acquihire in period 2 iff it has a high match. In period 2 a high-match firm 1 does an acquihire. A low-match firm 1 that received a $D$ shock, believes its competitor has a high match with probability $\lambda (1-r)(1-\gamma)$ and will do an acquihire if this is larger than $\lambda_A$.  Analogously, a low-match firm 1 receiving a $N$ shock will do an acquihire if $\lambda (1-\gamma(1-r)) \geq \lambda_A$. 
	\item Suppose firm 2 acquihired in period 1, so firm 1 is a low-match and will not do an acquihire moving second. Thus, firm 2 does an acquihire iff it has a high match.
	\item Suppose no acquihire took place in period 1. Then, both firms have a low match and this is commonly known, leading to no acquihires in period 2 either.
\end{itemize}

Moving to period 1, firm 2 does an acquihire iff it has a high match, knowing that firm 1 has a low match, since a high-match firm 1 would always do an acquihire. For a low-match firm 1, doing nothing yields
\begin{align*}
	\lambda\left(\underline{\Pi}_F^H(2-\gamma\delta) + \gamma\delta \Pi_F\right)+ (1-\lambda) 2 \Pi_F = \lambda(2-\gamma\delta)(\underline{\Pi}_F^H-\Pi_F)+ 2\Pi_F.
\end{align*}
The payoff of an acquihire depends on parameters and reads:
{\small{
\begin{itemize}
    \item $2( \bar{\Pi}_F^L-\pi_E)$ if $\lambda (1-r)(1-\gamma) >\lambda_A$;
    \item $(\bar{\Pi}_F^L-\pi_E)(2-\delta\gamma) - \delta\gamma\lambda (1-r)(1-\gamma)(\Pi_F -\underline{\Pi}_F^H) + \delta \gamma\Pi_F$ if $\lambda (1-\gamma(1-r))>\lambda_A > \lambda (1-r)(1-\gamma)$;
    \item $(\bar{\Pi}_F^L-\pi_E)(2-\delta) - \delta\lambda(1-\gamma)(\Pi_F -\underline{\Pi}_F^H) + \delta \Pi_F$ if $\lambda_A > \lambda (1-\gamma(1-r))$.
\end{itemize}
}}
\noindent In case 1, a low-match firm will always hoard talent in period 2. In case 2, it will hoard talent unless it receives a $D$ shock in an economic downturn. In case 3, it will hoard talent as long as the economy does not experience a downturn. Thus, comparing the total payoffs from doing nothing or an acquihire in period 1, the acquihire thresholds for period 1 read, respectively, 
\begin{align*}
	\lambda_A^1 &= \lambda_A\cdot \frac{2}{2-\gamma\delta},\\
	\lambda_A^2 &= \lambda_A\cdot \frac{2-\delta \gamma}{2- \delta \gamma - (1-r)(1-\gamma)\delta \gamma},\\
	\lambda_A^3 &= \lambda_A,
\end{align*}
so that $\lambda_A^1 \geq \lambda_A^2 \geq \lambda_A^3$.

\paragraph{Comparison.}
To compare hiring and layoffs, we need to consider three cases.

	\textit{Case 1:}  $\lambda (1-r)(1-\gamma) >\lambda_A$. Then, $\lambda \geq \lambda_A^1 = \lambda_A\frac{2}{2-\gamma\delta}$. Hence, a low-match firm will do an acquihire in both periods so no layoffs are observed, which is less than in the benchmark. Thus, irrespective of whether the economy hits a downturn, the entrepreneur is always employed without separation.
 
	\textit{Case 2:} $\lambda (1-\gamma(1-r))>\lambda_A >\lambda (1-r)(1-\gamma)$. Hence, $\lambda>\lambda_A^2$ so firm 1 will always do an acquihire in period 1. In period 2, firm 1 will maintain employment of the entrepreneur unless it receives shock $D$. Hence, the probability of a layoff will be $\delta\gamma$, which is larger than the benchmark layoff rate $l^*$. 
 %Moreover, the probability of transition to a permanent layoff will be $\delta(\gamma-\lambda P(D, N)),$ which is larger than $u^*$ iff $\frac{1-\lambda}{\lambda}>(1-r)(1-\gamma)$. 

\textit{Case 3:} $\lambda_A > \lambda (1-\gamma(1-r))$. If $\lambda>\lambda_A^3 $ then firm 1 will do an acquihire in period 1. Firm 1 will maintain employment in period 2 unless it was hit by a downturn and has a low match. Hence, the probability of observing a layoff is $\delta(\gamma+(1-\lambda)(1-\gamma))=\delta(1-\lambda(1-\gamma))$, which is larger than $l^*$. 
%The probability of observing a transition to a permanent layoff is
%$
%\delta([2\lambda-\lambda^2]\gamma-\lambda^2P(D, N)+(1-\lambda)^2)
%$
%which is larger than $u^*$.
If instead $\lambda<\lambda_A^3$, then we have no talent hoarding in either stage and the equilibrium is identical to the benchmark.

Finally, observe that the condition $\min\left\{\frac{\lambda_A}{\lambda}, \frac{1-\lambda}{\lambda}\right\} > (1-r)(1-\gamma)$ implies that we are either in Case 2 or 3. 
%Further, in Case 2, it ensures that we have more permanent layoffs than in the benchmark. 

\subsection{Proof of Proposition \ref{prop:people_tech}}

Suppose firm 1 has not made the acquisition. It follows from Assumption \ref{ass:A3} that firm 2 makes the acquisition if and only if it has a high match type, irrespective of its beliefs. Moving to stage 1, firm 1's beliefs are given by the prior belief. Suppose a high-match firm 1 acquired the startup. By Assumption \ref{ass:A4} a low-match firm 2 would not buy the technology part of the startup and by our tie-breaking assumption neither would a high-match firm 2. Then, by Assumption \ref{ass:A3} a high-match firm acquires the startup and keeps the technology. Suppose a low-match firm 1 acquired the startup. By our tie-breaking assumption, a low-match firm 2 would not buy the startup's technology. However, by Assumption \ref{ass:A4}, the technology part would be sold to a high-match firm 2. Anticipating this, the threshold for a low-match firm 1 to acquire the startup changes relative to the model in the main text. Formally, doing nothing yields
	\begin{align*}
		\Pi_F - \lambda \underline{\pi}_F^H,
	\end{align*}
	while making an acquisition yields
	\begin{align*}
		\Pi_F -\pi_E  +(1-\lambda)\bar{\pi}_F^L + \lambda(q + (1-\tau)\bar{\pi}_F^L - \tau \underline{\pi}_F^H),
	\end{align*}
	where
\begin{align*}
	q = \frac{\tau\left(\bar{\pi}_F^{H}+ \underline{\pi}_F^{L}+\bar{\pi}_F^{L}+\underline{\pi}_F^{H}\right)}{2},
\end{align*}
is the surplus-splitting sale price. Thus, an acquisition takes place whenever
\begin{equation*}
	\Pi_F -\pi_E  +(1-\lambda)\bar{\pi}_F^L + \lambda\left(\frac{\tau\left(\bar{\pi}_F^{H}+ \underline{\pi}_F^{L}+\bar{\pi}_F^{L}+\underline{\pi}_F^{H}\right)}{2} + (1-\tau)\bar{\pi}_F^L - \tau \underline{\pi}_F^H\right)\geq \Pi_F - \lambda \underline{\pi}_F^H,
\end{equation*}
which we can rearrange to the expression in the Proposition.

\subsection{Proofs of Propositions \ref{prop:woblocking} and \ref{prop:wblocking} } 

In what follows, we prove Proposition \ref{prop:wblocking}, which is a generalization of Proposition \ref{prop:woblocking}. We solve the game backward.
	
	\subparagraph{Stage 2}
	Observe that firm 2 believes with probability 1 that firm 1 is a low type, as otherwise, an acquihire would have taken place in stage 1. Firm 2 has three actions: doing nothing, making a bid that is accepted by both the entrepreneur and firm 1 and making a bid which is accepted only by the entrepreneur. To understand this, observe that the payoffs of firm 1 and the entrepreneur following a firm-1 investment of size $s_1$ at price $(p_1, d_1)$ read
	
	\begin{align*}
		\text{Firm 1}:~ & \Pi_F + s_1 \pi_E(s_1) - p_1-d_1\\
		\text{Entrepreneur}:~ & (1-s_1)\pi_E(s_1) + w(s_1) + p_1 + d_1.
	\end{align*}
	Hence,  firm 1 would try to block any bid $(p_2,d_2)$ resulting in a lower payoff than the above, while the entrepreneur would not accept any bid yielding a lower payoff than the above. Given these constraints, firm 2 will choose among three options: (i) do nothing and receive payoff $\Pi_F$; (ii) make the minimum bid such that both firm 1 and the entrepreneur accept, yielding payoff $\bar{\Pi}_F^\theta- \pi_E(s_1) - w(s_1)-\Pi_F + \underline{\Pi}_F^\theta$ for firm 2; or (iii) make a bid that only the entrepreneur accepts, risking a blocking attempt by firm 1. This third option provides an expected payoff for firm 2 of $\beta(s_1) \Pi_F + (1-\beta(s_1))(\bar{\Pi}_F^\theta- \pi_E(s_1)-w(s_1))$. 
	
	Let us consider a low-match firm 2 first. Comparing doing nothing with inducing only the entrepreneur to accept, we obtain that the latter move is is better for firm 2 whenever
	\begin{align*}
		\bar{\Pi}_F^L- \pi_E(s_1) -w(s_1) \geq \Pi_F.
	\end{align*}
	It follows from Assumptions \ref{ass:A1} and \ref{ass:A5} that this condition is not necessarily satisfied. Let $\hat{s}$ be the threshold above which this condition is satisfied.  Now let's compare doing nothing with inducing both for $s_1<\hat{s}$. We obtain that inducing both is better whenever
	\begin{align*}
		\bar{\Pi}_F^L- \pi_E(s_1) - w(s_1)-\Pi_F + \underline{\Pi}_F^L \geq \beta(s_1) \Pi_F + (1-\beta(s_1))(\bar{\Pi}_F^L- \pi_E(s_1)-w(s_1)),
	\end{align*}
	which, after rearrangement, can be rewritten as
	\begin{align}
	\label{partial-cutoff-1}
		\beta(s_1)(\bar{\Pi}_F^L-\Pi_F- \pi_E(s_1)-w(s_1)) \geq \Pi_F - \underline{\Pi}_F^L.
	\end{align}
	Since the right-hand side of \eqref{partial-cutoff-1} is positive and the left-hand side of it is negative for $s_1<\hat{s}$, this condition is never satisfied. Thus, for $s_1<\hat{s}$ the low type does nothing.

	Now let us compare firm 2 inducing only the entrepreneur versus inducing both to accept when $s_1\geq \hat{s}$. As above, inducing both is better whenever condition \eqref{partial-cutoff-1} holds. Since  the left-hand side of \eqref{partial-cutoff-1} is increasing in $s_1$, we define $s^L\in [\hat{s},1]$ as the threshold above which the condition is satisfied. Overall, a low-match firm 2 will take the following actions: for $s_1<\hat{s}$, do nothing; for $\hat{s}\leq s_1\leq s_L$, induces only the entrepreneur to accept; and for $s_1>s_L$, induce both the entrepreneur and firm $1$ to accept the offer.

	Next, we turn to the high-match firm 2. Comparing payoffs of doing nothing and inducing only the entrepreneur, it follows from Assumption 1 that the firm will always find the former option inferior. So we only need to compare the payoffs of inducing both versus only the entrepreneur. In particular, inducing both is optimal whenever 
	\begin{align*}
		\beta(s_1)(\bar{\Pi}_F^H-\Pi_F- \pi_E(s_1)-w(s_1)) \geq \Pi_F - \underline{\Pi}_F^H,
	\end{align*}
	which may be satisfied for $s_1$ above some threshold $s^H \in [0,1]$. Below this threshold, the high type will induce only the entrepreneur. Note that it is not clear whether $s^H$ or $s^L$ are bigger. 
	
	\subparagraph{Stage 1} 
	
	Observe that firm 1's belief about firm 2's type is given by the prior. Firm 1 can acquire different stakes $s_1$ which in turn may induce different responses from firm 2.  Specially, we have learned that a low-match firm 2 may do either nothing (N), induce only the entrepreneur (E), or induce both the entrepreneur and firm 1 to accept a bid (B).  As for a high-match firm 2, it may either do E or B. In what follows let $(A_1,A_2)$ denote the action profiles of the low- and high-match firm 2, e.g, $(N,B)$ means firm 2 does nothing when it is a low type, and it induces both to accept when its type is high. Let $\Delta (s_1)\equiv \pi_E(s_1)+w(s_1)-\pi_E(0)-w(0)$. The following are the firm-1 payoffs resulting from an acquisition of $s_1$ which induces the indicated firm-2 behavior:
	\begin{align*}
		(N,E) :&\; \lambda (1-(\beta(s_1))) \underline{\Pi}_F^H + (1-\lambda(1-\beta(s_1))) \Pi_F + \Delta (s_1)\\
		(N,B) :&\; \Pi_F + \Delta (s_1)\\
		(E,B) :&\; (1-(1-\lambda)(1-\beta(s_1))) \Pi_F+ (1-\lambda)(1-\beta(s_1)) \underline{\Pi}_F^L + \Delta (s_1)\\
		(B,B) :&\;  \Pi_F + \Delta (s_1)\\
		(E,E) :&\;\beta(s_1) \Pi_F+ (1-\beta(s_1)) (\lambda \underline{\Pi}_F^H  + (1-\lambda)\underline{\Pi}_F^L) + \Delta (s_1)\\
		(B,E) :&\;\lambda (1-(\beta(s_1))) \underline{\Pi}_F^H + (1-\lambda(1-\beta(s_1))) \Pi_F +\Delta (s_1)
	\end{align*}
	
	To illustrate how to calculate these payoffs, consider the case $(N,E)$, where the low-match firm 2 does nothing and the high-match induces only the entrepreneur to accept (while firm 1 would try to block such an acquihire attempt by firm 2). Observe that the lowest bid at which the entrepreneur is willing to sell a stake $s_1$ to firm 1 is $p_1 + d_1 = \pi_E(0) + w(0) - (1-s_1)\pi_E(s_1) - w(s_1)$. As we are considering the case $(N,E)$, so that a low-match firm 2 would do nothing, yielding a payoff of
	\begin{align}
	\label{partial-example-1}
		\Pi_F + s_1\pi_E(s_1) - (\pi_E(0) + w(0) - (1-s_1)\pi_E(s_1) - w(s_1))= \Pi_F + \Delta(s_1).
	\end{align}
	A high-match firm 2 would make a bid that firm 1 will try to block, succeeding with probability $\beta(s_1)$, yielding the following payoff to firm $1$
	\begin{align}
	\label{partial-example-2}
		\beta(s_1) \left(\Pi_F + \Delta(s_1)\right)+(1-\beta(s_1)) \left(\underline{\Pi}_F^H + \Delta(s_1)\right).
	\end{align}
	Adding \eqref{partial-example-1} and \eqref{partial-example-2} up while multiplying them with the probabilities $1-\lambda$ and $\lambda$, respectively, we obtain the expression in the above list.\footnote{Observe that the size of firm 1's investment $s_1$ is not the same across cases in the above list, as different investment sizes induce the different behaviors of firm 2.}  Finally, to complete the list, note that an acquihire gives a payoff $\bar{\Pi}_F^L - \pi_E(0)-w(0)$ to firm $1$ and doing nothing results in $\lambda \underline{\Pi}_F^H + (1-\lambda) \Pi_F$.
	
	To determine firm 1's equilibrium strategy, we need to further distinguish between three cases based on the values of the thresholds $s^H$, $s^L$, and $\hat{s}$. Recall that a low-match firm 2 will do nothing below $\hat{s}$, induce the entrepreneur between $\hat{s}$ and $s^L$ and potentially induce both above $s^L$, while a high-match firm 2 will induce only the entrepreneur below $s^H$ and may induce both above it. 
	
	\textit{Case 1:} $s^H \leq \hat{s} \leq s^L$. Note that this allows only for four types of firm-2 behavior following an investment. If $s_1\leq s^H$, then we have $(N,E)$, a low-match firm 2 does nothing and a high-match induces only the entrepreneur ; if $s^H <s_1 \leq \hat{s}$, then we have $(N,B)$, namely a low-match firm 2 does nothing and a high-match induces both; if $\hat{s}< s_1\leq s^L$, then a low-match firm 2 induces the entrepreneur and a high-match induces both, so we have $(E,B)$; and if $s_1 > s^L$, both types of firm 2 induce both, so we end up with $(B,B)$. Comparing these payoffs, we observe that an acquihire dominates investments inducing $(B,B)$ and $(E,B)$, while an investment inducing $(N,B)$ dominates an acquihire. Thus, the remaining actions are doing nothing, inducing $(N,E)$ or inducing $(N,B)$.
	\begin{align*}
		(N,E) :&\; \lambda (1-(\beta(s_1))) \underline{\Pi}_F^H + (1-\lambda(1-\beta(s_1))) \Pi_F + \Delta(s_1)\\
		(N,B) :&\; \Pi_F +\Delta(s_1)\\
		(N): & \; \lambda \underline{\Pi}_F^H + (1-\lambda) \Pi_F,
	\end{align*}
	where the investment necessary to induce $(N,B)$ is bigger than the one for $(N,E)$ but both are smaller than $\hat{s}$ so that $\Pi_F + \pi_E(s_1) + w(s_1)\geq \bar{\Pi}_F^L$. We note that: 
	\begin{itemize}
		\item As $\lambda\rightarrow 0$ doing nothing dominates both types of investment
		\item For $\lambda>\frac{\pi_E(0)+w(0)-\pi_E(s_1) - w(s_1) }{\Pi_F-\underline{\Pi}_F^H}$
		inducing $(N,B)$ dominates $(N)$
		\item Depending on parameters, $(N,E)$ or $(N,B)$ is the better investment, but for low enough $\lambda$ inducing $(N,E)$ is always better	
	\end{itemize}
	In summary, there is a threshold value for $\lambda$ below which doing nothing is best, then for larger $\lambda$ inducing $(N,E)$ is better, and for very large $\lambda$, depending on parameters, inducing $(N,B)$ may be best.

	\textit{Case 2:} $\hat{s} \leq s^H \leq  s^L$. Note that this allows only for four types of firm-2 behavior following an investment:  $(N,E)$, $(E,E)$, $(E,B)$ and $(B,B)$. Proceeding as above, we find that there is a threshold value for $\lambda$ below which doing nothing is best, then for larger $\lambda$ inducing $(N,E)$ is better, and for very large $\lambda$, depending on parameters, doing an acquihire may be best.
	
	\textit{Case 3:} $ \hat{s} \leq s^L\leq s^H $. Note that this allows only for four types of firm-2 behavior following an investment:  $(N,E)$, $(E,E)$, $(B,E)$ and $(B,B)$. Proceeding as above, we find that there is a threshold value for $\lambda$ below which doing nothing is best, then for larger $\lambda$ inducing $(N,E)$ is better, and for very large $\lambda$, depending on parameters, doing an acquihire may be best.

\section{Further Extensions}  \label{sec:app_extensions}

\setcounter{equation}{0}
\numberwithin{equation}{section}

\setcounter{proposition}{0}
\numberwithin{proposition}{section}

\renewcommand{\theproposition}{B\arabic{proposition}}

\setcounter{corollary}{0}
\numberwithin{corollary}{section}

\renewcommand{\thecorollary}{B\arabic{corollary}}

\setcounter{assumption}{0}
\numberwithin{assumption}{section}

\renewcommand{\theassumption}{B\arabic{assumption}}

\subsection{Simultaneous Moves}

First, consider the complete-information situation where the entrepreneur tries to maximize revenue from selling the startup to the two firms. Then, we obtain quite directly from Proposition 1 in \citet{JehielMoldovanuStacchetti1996} the following result for the case of a low-type and a high-type firm. 

\begin{proposition}
    The low-type firm buys the startup at price $p^L = \bar{\Pi}_F^L - \underline{\Pi}_F^H$ whenever $\bar{\Pi}_F^L+ \underline{\Pi}_F^L \geq\bar{\Pi}_F^H+ \underline{\Pi}_F^H$. Otherwise, the high-type firm buys the startup at price $p^H = \bar{\Pi}_F^H - \underline{\Pi}_F^L$.
\end{proposition}
The only thing to verify is that these prices exceed $\pi_E$. This can be seen from $p^H = \bar{\Pi}_F^H - \underline{\Pi}_F^L \geq \bar{\Pi}_F^H - \Pi_F > \pi_E$ (Assumption \ref{ass:A1}) and because in case of a (revenue-maximizing) sale to the low-type firm we must have $p^L \geq p^H$ and thus $p^L> \pi_E$. 

The application of Proposition 2 in \citet{JehielMoldovanuStacchetti1996} for the incomplete-information case is not as straightforward, as the frameworks do not quite match anymore. However, their Application 2 (p. 824) comes quite close and suggests that analagous conditions exist to determine when the low-type or the high-type firm succeeds in buying the startup.

\subsection{Uncertain Order of Moves} \label{app:unknownorder}

We consider a variation of our baseline model in which the order in which the firms move is privately drawn by nature with uniform probabilities. Hence, firm 1 is not necessarily moving first anymore. Specifically, when firm $i$ gets to interact with the startup,  it does not directly observe whether the other firm has already interacted with the startup. Still, firm $i$ can make a bid to acquihire the startup, which the entrepreneur can accept or reject. Ex-ante probabilities of the firms' private match types are unchanged and the payoffs following an acquisition, too. We obtain the following result.

\begin{proposition}
	Under Assumption \ref{ass:A1}, there exists a symmetric equilibrium in which all firms acquihire the startup at price $\pi_E$ independently of their type if 	
	\begin{align} \label{eq:unknownorder}
		\lambda \geq \lambda_{A,U}\equiv\frac{\pi_E+\underline{\Pi}^L_F-\bar{\Pi}_F^L}{\underline{\Pi}^L_F-\underline{\Pi}_F^H}.
	\end{align}
\end{proposition}
\begin{proof}
	Suppose that firm $j$ behaves as suggested in the proposition and that the entrepreneur accepts bids if and only if they are at least $\pi_E$. We consider the incentives of firm $i$. Given Assumption \ref{ass:A1}, it is a dominant strategy for a high-match firm $i$ to make a bid $\pi_E$ to do an acquihire whenever it has the chance to do so. Now, consider a low-match firm $i$. Given that firm $j$ would always do an acquihire, firm $i$ knows that it is moving first. Hence, doing nothing yields a payoff of $\lambda \underline{\Pi}_F^H + (1-\lambda) \underline{\Pi}_F^L$, while doing an acquihire yields $\bar{\Pi}_F^L - \pi_E$. Hence, doing the acquihire is optimal if and only if $\lambda \geq \lambda_{A,U}$, completing the proof. 
\end{proof}

The result shows that our result in Proposition  \ref{prop:AHonly} remains qualitatively unchanged when firms do not the order in which they move. In particular, talent hoarding continues to arise as long as high types are sufficiently likely.

\subsection{Emerging Competitor}

We consider a scenario in which the two incumbents do not coexist. Specifically, only one incumbent is active, while the other firm may enter with probability $q\in [0,1]$. If the second firm does not enter, and there is no acquihire, the first firm makes a profit $\Pi_M$. Acquihiring the startup yields a match-dependent profit of $\bar{\Pi}_M^H$ or $\bar{\Pi}_M^L$, where---mirroring Assumption \ref{ass:A1}---we assume $\bar{\Pi}^H_M>\Pi_M+ \pi_E>\bar{\Pi}^L_M$. If the second firm does enter, all profits align with those in the benchmark. 

\begin{proposition}
	Suppose Assumption \ref{ass:A1}  and $\bar{\Pi}^H_M>\Pi_M+ \pi_E>\bar{\Pi}^L_M$ hold, then  firm 1's behavior in any perfect Bayesian equilibrium (PBE) is uniquely specified. Namely, if firm 1 is a high match with the startup, it will pursue an acquihire; if it is a low match, it will pursue an acquihire if and only if
\begin{align}
    \lambda \geq \lambda_{A, M}\equiv \frac{\pi_E - q(\bar{\Pi}_F^L - \Pi_F)- (1-q)(\bar{\Pi}_M^L - \Pi_M)}{q(\Pi_F - \underline{\Pi}_F^H)}.
\end{align}
\end{proposition}

\vspace{2mm}

\begin{proof}
	Suppose firm 1 has not pursued an acquihire. If the competitor does not emerge, we are done and the monopolist's profit reads $\Pi_M$. If the competitor emerges, it follows from Assumption \ref{ass:A1} that, for any belief, firm 2 will pursue an acquihire if and only if it is a high type. Moving to stage 1, firm 1's belief is given by the prior. A high-match firm 1 will always pursue an acquihire by Assumption \ref{ass:A1}. A low-match firm 1  will do so whenever $q\bar{\Pi}_F^L+ (1-q)\bar{\Pi}_M^L - \pi_E \geq q(\lambda \underline{\Pi}_F^H + (1-\lambda)\Pi_F) + (1-q) \Pi_M$ or, equivalently, $\lambda\geq \lambda_{A, M}$.
\end{proof}

 \subsection{Surplus Sharing} \label{app:surplussharing}

We modify our baseline model to allow for arbitrary degrees of surplus sharing between the entrepreneur and the firm when an acquihire takes place. Thus, firms still move sequentially, but in case of an acquihire the entrepreneur receives a share $\xi \in [0,1]$ of the surplus. We define surplus here as the difference between the joint payoffs arising from an acquihire and the joint payoffs arising in the case of no acquihire. 

\begin{proposition}
Talent hoarding can happen if the following condition holds:
\begin{align*}
	\xi \leq \frac{\bar{\Pi}_F^L - \underline{\Pi}_F^H-\pi_E}{\bar{\Pi}_F^H - \Pi_F-\pi_E}.
\end{align*}
\end{proposition}
\begin{proof}
	We solve the game backward. Consider a high-match firm 2. The surplus resulting from an acquihire is given by $\bar{\Pi}_F^H - \Pi_F-\pi_E>0$ so that an acquihire takes place and the resulting payoffs for firm 2 and the entrepreneur read $\Pi_F + (1-\xi) (\bar{\Pi}_F^H - \Pi_F-\pi_E)$ and $\pi_E + \xi (\bar{\Pi}_F^H - \Pi_F-\pi_E)$, respectively. Consider a low-match firm 2. The surplus then reads $\bar{\Pi}_F^L - \Pi_F-\pi_E<0$ so that no acquihire takes place. 
	
	Moving to period 1, consider a high-match firm 1. The surplus resulting from an acquihire reads
	\begin{align*}
		&\bar{\Pi}_F^H - (\lambda \underline{\Pi}_F^H + (1-\lambda) \Pi_F) - ((1-\lambda)\pi_E + \lambda (\pi_E + \xi (\bar{\Pi}_F^H - \Pi_F-\pi_E)))\\
		&= (\bar{\Pi}_F^H - \Pi_F-\pi_E)(1-\lambda\xi)  +\lambda (\Pi_F - \underline{\Pi}_F^H)\geq 0,
	\end{align*}
so that an acquihire takes place.  Consider a low-match firm 1. The surplus resulting from an acquihire reads
\begin{align*}
	(\bar{\Pi}_F^H - \Pi_F-\pi_E)(1-\lambda\xi)  +\lambda (\Pi_F - \underline{\Pi}_F^H) + \bar{\Pi}_F^L - \bar{\Pi}_F^H.
\end{align*}
This expression is positive (implying that an acquihire is profitable) whenever
\begin{align*}
	\lambda \geq \lambda_{A,S} \equiv\frac{\pi_E+\Pi_F-\bar{\Pi}_F^L}{\Pi_F-\underline{\Pi}_F^H - \xi(\bar{\Pi}_F^H - \Pi_F-\pi_E)},
\end{align*}
for $\Pi_F-\underline{\Pi}_F^H - \xi(\bar{\Pi}_F^H - \Pi_F-\pi_E)>0$. We have $\lambda_{A,S} \leq 1$ whenever the condition stated in the proposition holds.
\end{proof}

\subsection{Dominant Firm}\label{sec:dominant_firm}

Consider a situation where instead of two symmetric firms, the industry is characterized by a dominant firm and a challenger firm. The firms now have different payoffs, which we denote $(\bar{\Pi}^H_D, \bar{\Pi}_D^L, \Pi_D, \underline{\Pi}_D^L , \underline{\Pi}_D^H )$ for the dominant firm and $(\bar{\Pi}^H_C, \bar{\Pi}_C^L, \Pi_C, \underline{\Pi}_C^L , \underline{\Pi}_C^H )$ for the challenger. We maintain Assumption \ref{ass:A1} for both firms, that is we assume that both (i) $\bar{\Pi}^H_F>\Pi_F+ \pi_E>\bar{\Pi}_F^L$ and (ii) $\Pi_F\geq \underline{\Pi}_F^L >\underline{\Pi}_F^H$ hold for each $F \in \{D,C\}$. We consider the setting as in Proposition \ref{prop:AHonly}, that is firms can either engage in acquihires or do nothing.

\begin{corollary}[Acquihires with a dominant firm] \label{cor:DominantFirm} $ $
	\begin{enumerate}[label=(\roman*)] 
		\item If either the dominant or the challenger firm moves second it engages in an aquihire if and only if it is the high match type. 
		\item If the dominant firm moves first, it engages in an acquihire if it is the high match type or if it is the low match type and the probability that the challenger firm is the high match type is
		\begin{align}
			\lambda>\lambda_{A,D}\equiv\frac{\pi_E+\Pi_D-\bar{\Pi}_D^L}{\Pi_D-\underline{\Pi}_D^H}.
		\end{align}
        \item If the challenger firm moves first, it engages in an acquihire if it is the high match type or if it is the low match type and the probability that the dominant firm is the high type is
		\begin{align}
			\lambda>\lambda_{A,C}\equiv\frac{\pi_E+\Pi_C-\bar{\Pi}_C^L}{\Pi_C-\underline{\Pi}_C^H}.
		\end{align}
	\end{enumerate}
\end{corollary}

The corollary follows directly from Proposition 1. An interesting question is under which conditions would the dominant firm be more prone to talent hoarding than the challenger firm, i.e., when is $\lambda_{A,C} > \lambda_{A,D}$? The following result gives a set of simple sufficient conditions.

\begin{proposition}\label{prop:DominantFirmSufficient}
	If the following two inequalities hold, then $\lambda_{A,C} > \lambda_{A,D}$:
	\begin{align*}
		\Pi_C-\underline{\Pi}_C^H & < \Pi_D-\underline{\Pi}_D^H ,\\
		\bar{\Pi}_C^L - \Pi_C &< \bar{\Pi}_D^L - \Pi_D.
	\end{align*}
\end{proposition}
\begin{proof}
From Corollary \ref{cor:DominantFirm} we have the threshold values 
		\begin{align*}
	\lambda_{A,D}\equiv\frac{\pi_E+\Pi_D-\bar{\Pi}_D^L}{\Pi_D-\underline{\Pi}_D^H}, \quad \text{and }\quad \lambda_{A,C}\equiv\frac{\pi_E+\Pi_C-\bar{\Pi}_C^L}{\Pi_C-\underline{\Pi}_C^H}.
\end{align*}
Note  that the conditions stated in the proposition imply
\begin{align*}
	\lambda_{A,D} = \frac{\pi_E+\Pi_D-\bar{\Pi}_D^L}{\Pi_D-\underline{\Pi}_D^H}<  \frac{\pi_E+\Pi_D-\bar{\Pi}_D^L}{\Pi_C-\underline{\Pi}_C^H}< \frac{\pi_E+\Pi_C-\bar{\Pi}_C^L}{\Pi_C-\underline{\Pi}_C^H}= \lambda_{A,C},
\end{align*}
completing the proof.
\end{proof}

Intuitively, the two inequalities above require that the dominant firm both stands to lose more if a startup is acquired by the challenger (maybe because the dominant firm has a larger market share, so it has more to lose) and stands more to gain by acquihiring the startup itself (a larger market share might be an explanation again, as any improvement could be offered to more consumers more rapidly). 

A simple specification that satisfies these inequalities is an ``equal proportional gain/loss''. Formally, let the profit functions be given by 
\begin{align*}
	\bar{\Pi}^H_D = H \Pi_D, \quad \bar{\Pi}_D^L = L \Pi_D, \quad \underline{\Pi}_D^L = \ell \Pi_D, \quad \underline{\Pi}_D^H = h \Pi_D, \\
	\bar{\Pi}^H_C = H \Pi_C, \quad  \bar{\Pi}_C^L= L \Pi_C, \quad  \underline{\Pi}_C^L  = \ell \Pi_C, \quad   \underline{\Pi}_C^H= h \Pi_C,   
\end{align*}
where $H>L>1$ and $1 \leq \ell > h \geq 0$. This implies that an acquihire (either own or by the competitor) has a proportionally equal effect on both the dominant firm and the challenger firm. As long as $\Pi_D > \Pi_C$ (i.e., absent any acquihire, the dominant firm has higher profits than the challenger), straightforward calculations show that the two inequalities of Proposition \ref{prop:DominantFirmSufficient} are satisfied and the dominant firm is more prone to acquihires than the challenger.

\subsection{Multiple Firms}\label{sec:multiple_firms}

Often, more than two firms are competing in a given market. In this extension, we thus allow for $n\geq 2$ firms competing in the same market. As in the baseline, firms may sequentially attempt to do an acquihire of the startup and each firm's match type $\theta\in \{L, H\}$ is an independent draw with identical probability $\Pr(\theta=H) = \lambda$.  In the absence of an acquihire, each firm makes profits $\Pi_F(n)$. If firm $i$ with match type $\theta$ makes an acquihire profits read $\bar{\Pi}_F^\theta(n)$ for firm $i$ and $\underline{\Pi}_F^\theta(n)$ for firms $j\neq i$.

To make things concrete, we focus on a Cournot oligopoly setting with $n$ firms and the following inverse demand function:
\begin{align}
    P(q_1,...,q_n)=a-b\cdot \sum_{i=1}^nq_i,
\end{align}
where $q_i$ indicates the quantity choice of firm $i\in \{1,...,n\}$  and $a, b>0$. We assume that the demand intercept $a$ is sufficiently large to ensure an interior solution. Let $c>0$ denote the constant marginal cost when no acquisitions occur. An acquihire by one firm is assumed to reduce its marginal cost to $c-\theta>0$, where $\theta\in \{H, L\}$  satisfies $H>L\geq 0$. 

We first establish that increasing competition eliminates talent hoarding in the limit: as $n\rightarrow +\infty$,  firms acquire startups if and only if it is efficient to do so.\footnote{As it will become clear, this result does not hinge on the Cournot specification; it holds as long as the benchmark profit $\Pi_F$ converges to zero as the number of firms increases.} To see this, take any one of the $n$ firms and suppose that its match value with the startup is $\theta$. It is clear that if
$\pi_E<\bar{\Pi}_F^{\theta}(n)-\Pi_F(n)$, this firm will acquire the startup whenever possible. In contrast, if
\begin{align}
\label{inefficient-acquihire}
    \pi_E>\bar{\Pi}_F^{\theta}(n)-\Pi_F(n),
\end{align}
acquihiring is inefficient. At the same time, a necessary condition for the firm to have incentives to do an acquihire is
\begin{align}
\label{incentive-acquihire}
    \bar{\Pi}_F^{\theta}(n)-\underline{\Pi}_F^{H}(n)>\pi_E.
\end{align}
Note that the LHS of \eqref{incentive-acquihire} is the maximum difference in the firm's profits between acquiring and not acquiring the startup. In our Cournot specification, as $n\rightarrow +\infty$, both $\Pi_F(n)$ and $\underline{\Pi}_F^{H}(n)$ converge to zero. Therefore, conditions \eqref{inefficient-acquihire} and \eqref{incentive-acquihire} cannot hold simultaneously when $n$ is sufficiently large. This implies that, whenever \eqref{inefficient-acquihire} holds, no firm with a match value $\theta$ will conduct an acquihire, therefore the talent hoarding problem cannot occur in equilibrium.

Next, we argue that the impact of competition on talent hoarding can be non-monotone. Note that when $n=2$, talent hoarding arises in equilibrium if and only if
\begin{align}
\label{talent-hoarding-2-firms}
    \pi_E\leq \lambda \left(\Pi_F(2)-\underline{\Pi}_F^H(2)\right)+\bar{\Pi}_F^L(2)-\Pi_F(2).
\end{align}
In comparison, suppose that $n=3$ and firm $1$ anticipates that firms $2$ and $3$ will acquire the startup when their match values are high.  Then, firm $1$ will be inclined to conduct an acquihire even when it draws a low match value, if the following condition holds:
\begin{align}
\label{talent-hoarding-3-firms}
    \pi_E\leq (2\lambda-\lambda^2) \left(\Pi_F(3)-\underline{\Pi}_F^H(3)\right)+\bar{\Pi}_F^L(3)-\Pi_F(3). 
\end{align}

Now, consider a parametric example with $\lambda=0.1$, $a=10$, $b=1$, $c=3$, $H=2$, $L=0$, and $\pi_E\in (0.534, 0.54)$. One can show that Assumption 1 holds. In addition, condition \eqref{talent-hoarding-2-firms} is violated but condition \eqref{talent-hoarding-3-firms} is satisfied. Thus, in the current example, talent hoarding will not occur in equilibrium when there are two firms. However, with three firms, talent hoarding will for sure arise in equilibrium.

\newpage
\bibliographystyle{ecta}
\bibliography{talenthoarding}

\end{document}